\documentclass[journal]{IEEEtran}
\usepackage{cite}

\usepackage[ruled,vlined]{algorithm2e}

\usepackage{mathtools, amssymb, amsthm}
\newtheorem{theorem}{Theorem}
\newtheorem{corollary}{Corollary}
\newtheorem*{corollary*}{Corollary}

\newtheoremstyle{remarkstyle}
  {3pt}   % space above
  {3pt}   % space below
  {\normalfont} % body font
  {}      % indent
  {\itshape} % heading font
  {:}     % punctuation after heading
  {0.5em} % space after heading
  {\thmname{#1}~\thmnumber{#2}}

\theoremstyle{remarkstyle}
\newtheorem{remark}{Remark}

\usepackage{tabularx, graphicx}
\usepackage[caption=false,font=footnotesize]{subfig}
\usepackage[table,dvipsnames]{xcolor}
\usepackage{pifont}  % For \ding

\usepackage{pgfplots}
\pgfplotsset{compat = newest}
\usetikzlibrary{positioning, calc, arrows, decorations.pathreplacing, patterns, mindmap, trees, shapes, backgrounds}

\newcommand{\SCALE}{0.25}   % global scale parameter
\newcommand{\SCALETWO}{0.58}   % global scale parameter
\usepackage{fontawesome5}
\usepackage{array}

\definecolor{palecgray}{gray}{0.95}      % pale cool gray panel
\definecolor{burntumber}{RGB}{138, 51, 36} % dark reddish-brown
\definecolor{chestnut}{HTML}{954535}     % warm brown 
\definecolor{cobalt}{RGB}{0, 71, 171}      % cobalt blue 
\definecolor{bazaar}{HTML}{98777B}       % muted mauve-brown
\definecolor{antiquebrass}{HTML}{CD9575} % "Antique Brass" 
\definecolor{ashgrey}{HTML}{B2BEB5}      % soft grey-green
\definecolor{mcrnyellow}{HTML}{FFBE00} % MCRN yellow
\definecolor{mcrnred}{HTML}{B92F0A}    % MCRN red
\definecolor{mcrnbrown}{HTML}{54191B}  % MCRN brown
\definecolor{mcrnblue}{HTML}{2660AE}   % MCRN blue
\definecolor{rowcolor1}{RGB}{230, 248, 230}
\definecolor{rowcolor2}{RGB}{255, 218, 195}

\title{Mitigating Heterogeneity-Induced Drift in Hierarchical Sign-Based Federated Learning}
\author{{Amirreza Kazemi, Seyed Mohammad Azimi-Abarghouyi,~\IEEEmembership{Member,~IEEE}, Gabor Fodor,~\IEEEmembership{Fellow,~IEEE}, \\and Carlo Fischione,~\IEEEmembership{Fellow,~IEEE}}

	\thanks{A. Kazemi, G. Fodor, and C. Fischione are with the School of Electrical Engineering and Computer Science, KTH Royal Institute of Technology, Stockholm, Sweden (e-mails: \{seykaz, gaborf, carlofi\}@kth.se). S. M. Azimi-Abarghouyi is with the Department of Electrical Engineering, Chalmers University of Technology, Gothenburg, Sweden (e-mail: azimimo@chalmers.se). This research has received funding from the European Union's Horizon Europe research project BATTwin under grant agreement No. 101137954.}}
    
\begin{document}
\bstctlcite{IEEEexample:BSTcontrol}
\maketitle
\begin{abstract}
Hierarchical federated learning (HFL) is well suited for large-scale wireless and Internet of Things systems, where devices communicate with nearby edge servers before reaching the cloud. In these environments, uplink bandwidth and latency impose strict communication constraints, making aggressive gradient compression essential. One-bit sign-based stochastic gradient descent methods provide an attractive solution in flat federated settings, but their behavior in hierarchical edge--cloud architectures remains insufficiently understood, especially under inter-cluster data heterogeneity. To address this gap, we develop a sign-based HFL framework in which devices transmit binary stochastic-gradient signs to edge servers, edge servers apply majority voting, and the cloud periodically aggregates edge models. Our analysis reveals that inter-cluster heterogeneity induces a persistent bias term in the convergence bound, reflecting the drift of edge models toward local objectives. This term cannot be removed by increasing the number of training rounds or by tuning standard hyperparameters alone. We therefore propose \(\mathtt{DC\text{-}HierSignSGD}\), a drift-corrected sign-based HFL algorithm in which devices apply a cloud-assisted gradient correction before taking the sign. We show that this pre-sign correction mitigates the non-vanishing heterogeneity-induced bias while preserving binary device--edge communication during the repeated local sign-update steps. Experiments under severe inter-cluster heterogeneity demonstrate that \(\mathtt{DC\text{-}HierSignSGD}\) improves the stability and accuracy of sign-based HFL and achieves performance comparable to full-precision hierarchical SGD with substantially lower device--edge communication.
\end{abstract}

\begin{IEEEkeywords}
Hierarchical federated learning, edge--cloud networks, SignSGD, data heterogeneity, communication-efficient learning.
\end{IEEEkeywords}

\section{Introduction}

Federated learning (FL) allows distributed devices to collaboratively train a global model without sharing raw data, reducing privacy risks. 
This approach is particularly advantageous when data are expensive to collect or challenging to aggregate. Moreover, FL enhances computational efficiency by allowing multiple devices to train the model concurrently \cite{mcmahan2017communication}.

Despite its advantages, the standard FL paradigm struggles with scalability. As the number of participating devices grows, communication rounds become progressively slower, uplink congestion increases, and stragglers increasingly dominate the overall round latency. To address these issues, numerous works in FL have introduced modified algorithms aimed at mitigating the effects of the aforementioned challenges \cite{reisizadeh2020fedpaq,reisizadeh2022straggler,raftopoulou2024agent,mahmoudi2025accelerating}. 

A recent structural evolution of the standard FL framework is hierarchical federated learning (HFL), in which devices send their updates to intermediate edge servers rather than directly to the cloud server \cite{castiglia2021multi,liu2022hierarchical,azimi2025hierarchical,azimi2024hierarchical}. In its typical form, HFL employs a two-stage aggregation process: edge servers first combine updates from their associated devices, and the cloud then aggregates the outputs of multiple edge servers to obtain the global model. By introducing this intermediate layer, HFL has the potential to reduce communication load, improve scalability, and enable more efficient distributed learning in large-scale networks.

Notwithstanding its benefits, HFL continues to face a critical bottleneck involving uplink communication between devices and edge servers. Transmitting full-precision gradients or model updates is costly, especially over noisy or bandwidth-constrained communication links where frequent two-way transmissions take place. As a result, communication-efficient learning techniques have become indispensable for HFL systems.
Most existing studies in FL and HFL improve communication efficiency by applying general quantization to model parameters or gradients during uplink transmission.

While such schemes reduce communication costs, they still require multiple bits per model dimension, which can be prohibitive due to the scale of modern learning models, even with a modest number of devices. This limitation typically occurs in settings constrained by bandwidth, latency, or energy, such as large-scale Internet of Things deployments with severely limited uplink rates~\cite{bonawitz2019towards, wang2019adaptive}, wireless edge-learning systems under stringent bandwidth budgets~\cite{yang2020federated, zhu2020broadband}, and distributed sensor networks, such as in fault prediction in factories (e.g., for electrical battery manufacturing processes or industrial automation processes), where nodes can transmit only a few bits per reporting cycle~\cite{fang2020energyefficient}. The challenge is further amplified in emerging applications involving large language models, where adding even a single extra bit per dimension substantially increases the communication load, rendering multi-bit schemes ineffective and making one-bit, sign-based methods particularly attractive.

% \begin{table*}[t]
% \centering
% \caption{Comparison of prominent existing works in HFL}
% \label{tab:literature_comparison}
% \begin{tabular}{|c|c|c|c|c|} % You can adjust 10cm as needed
% \hline
% \textbf{Work} & \textbf{Convergence analysis} & \textbf{Communication efficiency} & \textbf{Inter-cluster heterogeneity} \\
% \hline
% \cite{liu2022hierarchical}  &  \cmark  &  \cmark   &  \xmark\\
% \hline 
% \cite{azimi2025hierarchical}  &  \xmark  &  \xmark  &   \xmark   \\
% \hline
% \cite{azimi2024hierarchical} & \cmark  & \cmark  &  \cmark \\ 
% \hline
% \cite{jiang2024convergence} &  \cmark  &  \xmark  &   \cmark  \\
% \hline
% \cite{liu2023adaptive} & \cmark  &  \cmark & \xmark  \\  
% \hline
% This work & \cmark  &  \cmark & \cmark  \\  
% \hline
% \end{tabular}
% \end{table*}

A promising way to achieve this goal is to apply aggressive compression to the transmitted updates. In this context, sign-based stochastic gradient descent (SignSGD) is especially appealing because each device sends only the coordinate-wise signs of its stochastic gradients rather than full-precision values. This reduces the uplink payload of each device from \(d\) floats to $d$ bits ($d$ being the dimension of the model).
To appreciate the magnitude of this reduction, consider a baseline in which
devices transmit full-precision 32-bit gradient values or model updates.
Matching the communication budget of SignSGD with a sparsifier would require
retaining only about \(1/32\) of the entries, i.e., roughly \(3\%\), even before accounting for the additional cost of transmitting the support pattern. Such extreme compression is prone to instability and often degrades convergence due to the significant variance it introduces~\cite{liu2022hierarchical,alistarh2017qsgd}, as it discards essential information contained in the entries. But the directional information of gradients can be somewhat preserved even when reduced to a single bit per coordinate. This makes sign-based gradient communication particularly well suited for federated systems with severe communication constraints.

% Moreover, SignSGD can be viewed as a form of quantization possessing a fundamentally different and unique characterization that cannot be captured by the general or specific quantization schemes widely studied in many FL works \cite{bernstein2018signsgd}. In particular, existing quantization-based methods typically rely on simple statistical characterizations, such as modeling quantization error through its mean and variance, which makes their incorporation into convergence analysis relatively straightforward, since the quantization error appears as an additive noise term. In contrast, SignSGD requires a highly delicate, integer-based analysis, as its behavior is governed by sign operations and majority-vote mechanisms rather than additive quantization noise. Such treatment necessitates a new theoretical and analytical framework in HFL that is fundamentally different from those used in previous works.

Early works studied SignSGD in the classical distributed stochastic gradient descent (SGD) context, showing that even though the sign is a biased operator, convergence (in homogeneous or well-controlled settings) to stationary points is feasible under suitable assumptions \cite{bernstein2018signsgd, karimireddy2019error, safaryan2021stochastic}. More recent advances improve these results via variance reduction to tighten convergence bounds in nonconvex settings~\cite{jiang2024efficient, chzhen2023signsvrg}. Moreover, momentum variants of SignSGD have been shown to enjoy convergence under weaker smoothness conditions~\cite{sun2023momentum}. However, nearly all of the literature on SignSGD and especially its variants assumes the standard flat (device--server) communication topology, while in many real-world settings, networks naturally follow multi-layer hierarchical architectures \cite{wang2021potential,dhillon2012modeling}.

In hierarchical networks, a natural question arises: how should sign-based learning be designed when binary device-level updates are first aggregated at edge servers and then propagated through cloud-level model averaging? This question is nontrivial because inter-cluster heterogeneity can cause edge models to drift toward their local objectives, making naive sign-based aggregation insufficient. Unlike the flat distributed setting considered in \cite{SignsgdFV}, where heterogeneity is mainly reflected through unequal mini-batch sizes and different sign reliabilities, hierarchical settings introduce gradient dissimilarity across edge-level objectives. Such mismatch creates systematic edge drift, which cannot be removed by reweighting sign votes alone. Hence, addressing this issue requires a rigorous formulation that captures the interaction among one-bit updates, majority voting, cloud aggregation, and gradient dissimilarity across edge clusters. Developing such a framework, analyzing its convergence behavior, and introducing a drift-correction mechanism constitute the main goals of this work.

% However, realizing this potential requires a precise convergence analysis that reflects real-world operating conditions. Providing such an analysis is the main goal of this paper.

\subsection{Related Literature}

In contrast to the extensive literature on standard FL, research on hierarchical FL remains comparatively limited. For example, \cite{liu2022hierarchical} studies a quantized HFL algorithm, establishes a tight convergence bound for nonconvex objectives, and derives system-design insights such as adaptive aggregation intervals and device--edge association strategies. However, its convergence analysis does not fully address non-IID data distributions. Another study, \cite{azimi2025hierarchical}, develops an HFL algorithm from a different optimization perspective, although its convergence guarantees are limited to a restricted setting. Quantized communication under heterogeneous data is considered in \cite{azimi2024hierarchical}, which analyzes convergence and identifies aggregation intervals that improve communication efficiency and learning accuracy. More recent extensions to multi-layer HFL with layer-specific quantization generalize the theory to deeper aggregation graphs and characterize how local iteration counts and quantization parameters should scale~\cite{azimi2025multi}. Separately, \cite{jiang2024convergence} develops a unified convergence framework for HFL under partial participation and data heterogeneity, without considering quantized communication. Beyond quantization, model pruning has been used to reduce the number of transmitted parameters by eliminating redundant entries~\cite{liu2023adaptive}. 

Several works have also combined HFL with over-the-air computation, showing that analog aggregation and hierarchical clustering can improve scalability and robustness to interference and data heterogeneity in wireless networks~\cite{aygun2024over,azimi2024scalable,azimiOTA2024hierarchical}. Context-aware and scheduling-driven frameworks further improve training by dynamically managing device participation and communication resources~\cite{wen2022joint,zhang2024device}.
% Other studies have examined HFL in communication-constrained networked systems. For instance, unmanned aerial vehicle-assisted HFL has been analyzed under wireless channel effects using an unbiased aggregation mechanism \cite{zhagypar2025uav}, while handover-aware HFL has been investigated for open radio access network-based next-generation mobile networks~\cite{singh2025user}. 
Other works have studied HFL from a network-system perspective. User--edge association under statistical and network-topology constraints is studied in~\cite{mhaisen2022optimal}, and UAV-assisted two-tier HFL architectures are considered in~\cite{wang2024uav}.
Practical challenges such as unbalanced edge regions are investigated in~\cite{xu2024adaptive}, while loss-based heterogeneity in wireless HFL systems is studied in~\cite{ye2024fedhelo}. 
Heterogeneity-aware client association and staleness control have also been proposed to improve convergence in practical HFL deployments \cite{wu2023hiflash}. More recently, multi-timescale gradient correction has been introduced to mitigate model drift caused by data heterogeneity across different hierarchical levels~\cite{fang2024hierarchical}.

Despite these advances, sign-based methods remain largely unexplored in HFL, even though they offer substantial communication savings. In particular, their convergence behavior is not well understood when edge servers hold heterogeneous data distributions, where local sign updates may induce systematic drift from the global objective.

% While these works introduce valuable methods for HFL, they also exhibit limitations tied to key challenges such as incomplete convergence analysis, restricted bandwidth, and not considering data heterogeneity.

% A recurring issue is that data heterogeneity, especially across devices linked to different edge servers, adds another layer of difficulty, since in realistic deployments devices linked to different edge servers often hold substantially different data distributions.
%\vspace{-4pt}

\subsection{Contributions}

To our knowledge, this is the first work to develop a heterogeneity-aware sign-based learning framework for hierarchical edge--cloud networks. The proposed framework explicitly accounts for the interaction among binary device--edge updates, edge-level majority voting, periodic cloud aggregation, and inter-cluster gradient dissimilarity. Our main contributions are outlined below.

\begin{itemize}
\smallskip
    \item
    We develop a sign-based HFL framework for hierarchical edge--cloud networks, where devices transmit only stochastic-gradient signs during local training, edge servers aggregate them through majority vote, and the cloud periodically averages the resulting edge models. This provides a communication-efficient baseline for studying sign-based learning in hierarchical edge--cloud architectures.

\smallskip
    \item
    We provide a nonconvex convergence analysis showing how local aggregation, cloud aggregation, stochastic-gradient noise, and inter-cluster gradient dissimilarity affect convergence. The analysis reveals that, unlike HFL schemes based on full-precision SGD updates or their conventionally quantized variants, naive sign-based HFL suffers from a persistent heterogeneity-induced drift term, which cannot be removed by increasing the number of global rounds or by tuning the standard hyperparameters alone.

\smallskip
    \item
    To counter the effect of edge-level drift, we introduce \(\mathtt{DC\text{-}HierSignSGD}\), a drift-corrected sign-based HFL algorithm in which devices apply a cloud-assisted gradient correction before taking the sign. We analyze its convergence behavior, extend the result to the majority-vote setting, and show experimentally that the correction improves the stability and accuracy of sign-based HFL under inter-cluster heterogeneity while preserving binary device--edge communication during local training.
\end{itemize}

% \medskip
% The rest of the paper is organized as follows. Section~\ref{sec:problem} introduces the HFL framework, formalizes the global and edge-level objectives, and discusses the communication constraints motivating sign-based optimization. Section~\ref{sec:algorithm} presents the baseline sign-based HFL algorithm and its convergence analysis. Section~\ref{sec:DC-algorithm} develops a drift-corrected solution to mitigate the inter-cluster heterogeneity effect observed in the baseline method. Section~\ref{sec:simulations} presents numerical results evaluating the convergence behavior, stability, and comparative performance of the proposed method.
% Finally, Section~\ref{sec:conclusion} concludes the main results of the paper.

%\medskip
\textit{Notation:} Throughout the text,  we use bold lowercase letters and italic letters to indicate vectors and scalars, respectively. 
$[\mathbf{a}]_i$ is the $i$th element of $\mathbf{a}$.
The operator $\operatorname{sgn}(\cdot)$ represents the element-wise sign function.
A vector norm and its dual are denoted by $\|\cdot\|$ and $\|\cdot\|_\ast$, respectively. 
The inner product of two vectors is expressed by $\langle \cdot, \cdot \rangle$.
The expectation operator is denoted by $\mathbb{E}\{\cdot\}$. For any function \(f\), \(\nabla f\) indicates its gradient. The hat notation $\hat{({\cdot})}$ denotes an estimate of a given variable. Also, see Table \ref{tab:notation} for a more detailed summary of the notation and symbols used throughout the paper.

\begin{table}[t]
\caption{Summary of important notation}
\label{tab:notation}
\centering

\rowcolors{2}{rowcolor1}{palecgray}
\renewcommand{\arraystretch}{1.35}

% First col: fixed-ish width for symbols
% Second col: X = stretchable description column
\begin{tabularx}{0.97\columnwidth}{
  >{\centering\arraybackslash}m{1.4cm}
  >{\raggedright\arraybackslash}X
}

% --- HEADER ROW (both cells centered in the blue bar) ---
\rowcolor{cobalt!80}
\multicolumn{1}{c}{\color{white}\textbf{Symbol}} &
\multicolumn{1}{c}{\color{white}\textbf{Description}} \\
% --------------------------------------------------------

$Q$ & Number of edge servers \\
$\mathcal{V}^q$ & Set of devices managed by edge server $q$ \\
$\mathcal{D}_{qk}$ & Dataset belonging to device $k$ of cluster $q$ \\
$|\mathcal{D}_{qk}|$ & Size of the local dataset \\
$D_q$ & Number of data samples belonging to cluster $q$ \\
$N$ & Total number of samples in the hierarchical network \\
$d$ & Dimension of the model parameter vector \\
$\mathbf{w}$ & Global model parameter vector \\
$\mathcal{L}(\mathbf{w}; \boldsymbol{\xi})$ & Loss function for a single sample vector $\boldsymbol{\xi}$ \\
$f_{qk}(\mathbf{w})$ & Local loss function of device $k$ of cluster $q$ \\
$\mathcal{F}_q(\mathbf{w})$ & Loss function at edge server $q$ \\
$\mathcal{F}(\mathbf{w})$ & Global loss function \\
$T_G$ & Number of global rounds \\
$T_E$ & Number of local steps per global round \\
$t$ & Index of global rounds \\
$\tau$ & Index of local steps \\
$\mathbf{w}^{(t)}$ & Global model at iteration $t$ \\
$\mathbf{v}_q^{(t,\tau)}$ & Edge model $q$ at iteration $(t,\tau)$ \\
$\hat{\mathbf{g}}_{qk}^{(t,\tau)}$ & Stochastic gradient \\
$\mu$ & Step-size \\
$B$ & Batch-size \\
$L$ & The smoothness constant \\
$\sigma^2$ & Gradient component variance bound \\
$\zeta$ & Gradient dissimilarity constant \\
$\mathbf{c}_q^{(t)}$, $\mathbf{c}^{(t)}$ & Edge-level and global gradient anchors, respectively \\
$\boldsymbol{\delta}_q^{(t)}$ & Drift-correction vector at edge \(q\) \\
$\widetilde{\mathbf{s}}_q^{(t,\tau)}$ & Corrected sign vector \\
$\rho$ & Correction-strength parameter
\end{tabularx}
\end{table}

\section{Hierarchical Edge--Cloud FL} \label{sec:problem}
The hierarchical structure typically consists of two layers. The first (or top) layer enables communication between the cloud server and the edge servers, which function as intermediate aggregators. This layer is responsible for transmitting model parameters for the purpose of cloud aggregation.
The second layer establishes connections between edge devices and their corresponding edge servers. Over the communication channel in this layer, devices exchange privacy-preserving updates with their edge servers, which then update the local models accordingly. This process mirrors the behavior of a conventional FL framework. We assume reliable, high-capacity edge--cloud backhaul links, so their physical-layer impairments are not modeled. Nevertheless, edge--cloud synchronization is not cost-free, as each cloud aggregation incurs backhaul traffic, coordination overhead, and additional exposure of intermediate model updates. Thus, even though device--edge uplink is the main bottleneck, avoiding cloud synchronization after every local step remains beneficial.

\begin{figure}[t]
  \centering
  % \begin{figure}[!t]
%\centering
\scalebox{0.47}{
\begin{tikzpicture}[x=0.65pt,y=0.65pt]

% ==== Cluster 1 ====
\begin{scope}[shift={(3cm,0cm)}, scale=0.5]
    \begin{scope}[shift={(0cm,0cm)}]   \input{Tikz/Device_1} \end{scope}
    \begin{scope}[shift={(2cm,0cm)}]   \input{Tikz/Device_2} \end{scope}
    %\begin{scope}[shift={(0cm,-3cm)}]  \input{Tikz/Device_3} \end{scope}
    \begin{scope}[shift={(0cm,-3cm)}]  \input{Tikz/Device_K} \end{scope}
    \node[rotate=70] at (4.3cm,4.5cm) {\huge $.\,.\,.\,.\,.$};
    \begin{scope}[shift={(-7.2cm,0.6cm)}, scale=2.39]
    %Shape: Polygon Curved [id:ds9690632654418511] 
    \draw [dashed, ultra thick, color=mcrnbrown]   (139.86,43.29) .. controls (159.86,33.29) and (284.86,40.29) .. (264.86,60.29) .. controls (244.86,80.29) and (241.86,106.29) .. (261.86,136.29) .. controls (281.86,166.29) and (154.86,169.29) .. (134.86,139.29) .. controls (114.86,109.29) and (119.86,53.29) .. (139.86,43.29) -- cycle ;
    \end{scope}
    \draw (150,90) node {Cluster $1$}; 
\end{scope}

% ==== Cluster 2 ====
\begin{scope}[shift={(8cm,0cm)}, scale=0.5]
    \begin{scope}[shift={(0cm,0cm)}]   \input{Tikz/Device_1} \end{scope}
    \begin{scope}[shift={(2cm,0cm)}]   \input{Tikz/Device_2} \end{scope}
    %\begin{scope}[shift={(0cm,-3cm)}]  \input{Tikz/Device_3} \end{scope}
    \begin{scope}[shift={(0cm,-3cm)}]  \input{Tikz/Device_K} \end{scope}
    \node[rotate=70] at (4.3cm,4.5cm) {\huge $.\,.\,.\,.\,.$};
    \begin{scope}[shift={(-7.2cm,0.6cm)}, scale=2.39]
    %Shape: Polygon Curved [id:ds9690632654418511] 
    \draw [dashed, ultra thick, color=mcrnbrown]   (139.86,43.29) .. controls (159.86,33.29) and (284.86,40.29) .. (264.86,60.29) .. controls (244.86,80.29) and (241.86,106.29) .. (261.86,136.29) .. controls (281.86,166.29) and (154.86,169.29) .. (134.86,139.29) .. controls (114.86,109.29) and (119.86,53.29) .. (139.86,43.29) -- cycle ;
    \end{scope}
    \draw (150,90) node {Cluster $2$}; 
\end{scope}

\draw (575,130) node {\Huge $\cdots\cdots$};

% ==== Cluster Q ====
\begin{scope}[shift={(15cm,0cm)}, scale=0.5]
    \begin{scope}[shift={(0cm,0cm)}]   \input{Tikz/Device_1} \end{scope}
    \begin{scope}[shift={(2cm,0cm)}]   \input{Tikz/Device_2} \end{scope}
    %\begin{scope}[shift={(0cm,-3cm)}]  \input{Tikz/Device_3} \end{scope}
    \begin{scope}[shift={(0cm,-3cm)}]  \input{Tikz/Device_K} \end{scope}
    \node[rotate=70] at (4.3cm,4.5cm) {\huge $.\,.\,.\,.\,.$};
    \begin{scope}[shift={(-7.2cm,0.6cm)}, scale=2.39]
    %Shape: Polygon Curved [id:ds9690632654418511] 
    \draw [dashed, ultra thick, color=mcrnbrown]   (139.86,43.29) .. controls (159.86,33.29) and (284.86,40.29) .. (264.86,60.29) .. controls (244.86,80.29) and (241.86,106.29) .. (261.86,136.29) .. controls (281.86,166.29) and (154.86,169.29) .. (134.86,139.29) .. controls (114.86,109.29) and (119.86,53.29) .. (139.86,43.29) -- cycle ;
    \end{scope}
    \draw (150,90) node {Cluster $Q$}; 
\end{scope}

%== Central Server ====
\begin{scope}[shift={(7.7cm,5cm)}]
\input{Tikz/Central_server}
\end{scope}

% ==== Mini NNs: global aggregation (top-right) ====
\begin{scope}[scale=2,shift={(3.3cm,5.75cm)}]
  \tikzstyle{neuron}=[circle, draw=black, line width=0.4pt, minimum size=5pt, inner sep=0pt]
  \tikzstyle{input}=[neuron, fill=bazaar]
  \tikzstyle{hidden}=[neuron, fill=antiquebrass]
  \tikzstyle{output}=[neuron, fill=ashgrey]
  \node at (15,50) {\color{mcrnblue} Global Aggregation};
  \node[input] (I1) at (0,10) {};
  \node[input] (I2) at (0,25) {};
  \node[input] (I3) at (0,40) {};
  \node[hidden] (H1) at (15,15) {};
  \node[hidden] (H2) at (15,25) {};
  \node[hidden] (H3) at (15,35) {};
  \node[output] (O1) at (30,20) {};
  \node[output] (O2) at (30,30) {};
  \foreach \i in {1,2,3} {
    \foreach \h in {1,2,3} { \draw[line width=0.3pt, gray!70] (I\i) -- (H\h); }
    \draw[line width=0.3pt, gray!70] (H\i) -- (O1);
    \draw[line width=0.3pt, gray!70] (H\i) -- (O2);
  }
\end{scope}

% ==== Mini NNs: edge aggregation ====
\begin{scope}[scale=2,shift={(0.6cm,3.55cm)}]
  \tikzstyle{neuron}=[circle, draw=black, line width=0.4pt, minimum size=5pt, inner sep=0pt]
  \tikzstyle{input}=[neuron, fill=bazaar]
  \tikzstyle{hidden}=[neuron, fill=antiquebrass]
  \tikzstyle{output}=[neuron, fill=ashgrey]
  \node at (15,50) { Edge Aggregation};
  \node[input] (I1) at (0,10) {};
  \node[input] (I2) at (0,25) {};
  \node[input] (I3) at (0,40) {};
  \node[hidden] (H1) at (15,15) {};
  \node[hidden] (H2) at (15,25) {};
  \node[hidden] (H3) at (15,35) {};
  \node[output] (O1) at (30,20) {};
  \node[output] (O2) at (30,30) {};
  \foreach \i in {1,2,3} {
    \foreach \h in {1,2,3} { \draw[line width=0.3pt, gray!70] (I\i) -- (H\h); }
    \draw[line width=0.3pt, gray!70] (H\i) -- (O1);
    \draw[line width=0.3pt, gray!70] (H\i) -- (O2);
  }
\end{scope}

% ==== Mini NNs: Local Training  ====
\begin{scope}[scale=2,shift={(0.4cm,1cm)}]
  \tikzstyle{neuron}=[circle, draw=black, line width=0.4pt, minimum size=5pt, inner sep=0pt]
  \tikzstyle{input}=[neuron, fill=bazaar]
  \tikzstyle{hidden}=[neuron, fill=antiquebrass]
  \tikzstyle{output}=[neuron, fill=ashgrey]
  \node at (15,50) {\color{mcrnbrown} Local Training};
  \node[input] (I1) at (0,10) {};
  \node[input] (I2) at (0,25) {};
  \node[input] (I3) at (0,40) {};
  \node[hidden] (H1) at (15,15) {};
  \node[hidden] (H2) at (15,25) {};
  \node[hidden] (H3) at (15,35) {};
  \node[output] (O1) at (30,20) {};
  \node[output] (O2) at (30,30) {};
  \foreach \i in {1,2,3} {
    \foreach \h in {1,2,3} { \draw[line width=0.3pt, gray!70] (I\i) -- (H\h); }
    \draw[line width=0.3pt, gray!70] (H\i) -- (O1);
    \draw[line width=0.3pt, gray!70] (H\i) -- (O2);
  }
\end{scope}

% ==== Uplink (mcrnred arrows) — static ====
\draw[very thick, color=mcrnred, dashed, -latex] (200,385) -- (400,498);
\draw[very thick, color=mcrnred, dashed, -latex] (419,388) -- (430,498);
\draw[very thick, color=mcrnred, dashed, -latex] (725,388) -- (460,498);
\node[rotate=30] at (270,450) {\color{mcrnred} $\substack{\text{\normalsize {Uplink}} \\ \text{\normalsize {Transmission}}}$};

% % ==== Downlink (mcrnblue arrows) — static ====
 \draw[very thick, color=mcrnblue, -latex] (415,498) -- (215,385);
 \draw[very thick, color=mcrnblue, -latex] (438,495) -- (427,385);
 \draw[very thick, color=mcrnblue, -latex] (480,498) -- (745,388);
  \node[rotate=-20] at (625,457) {\color{mcrnblue} $\substack{\text{\normalsize {Broadcast}} \\[0.2em] \text{\normalsize {Model}}}$};

% % ==== Mini NNs: local training (top) ====
% \begin{scope}[scale=2,shift={(-1cm,3cm)}]
%   \tikzstyle{neuron}=[circle, draw=black, line width=0.4pt, minimum size=5pt, inner sep=0pt]
%   \tikzstyle{input}=[neuron, fill=bazaar!50]
%   \tikzstyle{hidden}=[neuron, fill=antiquebrass!50]
%   \tikzstyle{output}=[neuron, fill=ashgrey!50]
%   \node at (15,50) {\color{mcrnblue} Local Training};
%   \node[input] (I1) at (0,10) {};
%   \node[input] (I2) at (0,25) {};
%   \node[input] (I3) at (0,40) {};
%   \node[hidden] (H1) at (15,15) {};
%   \node[hidden] (H2) at (15,25) {};
%   \node[hidden] (H3) at (15,35) {};
%   \node[output] (O1) at (30,20) {};
%   \node[output] (O2) at (30,30) {};
%   \foreach \i in {1,2,3} {
%     \foreach \h in {1,2,3} { \draw[line width=0.3pt, gray!70] (I\i) -- (H\h); }
%     \draw[line width=0.3pt, gray!70] (H\i) -- (O1);
%     \draw[line width=0.3pt, gray!70] (H\i) -- (O2);
%   }
% \end{scope}

%%%%%%%%%%%%%%% Edge server 1
\begin{scope}[scale=1.2, shift={(1.0cm,9cm)}]
  \input{Tikz/Cell_tower}
  \node at (90,-90) {\color{mcrnred} $\substack{\text{\normalsize {Edge}} \\ \text{\normalsize {Server 1}}}$};
\end{scope}
%%%%%%%%%%%%%%% Edge server 2
\begin{scope}[scale=1.2, shift={(5.1cm,9cm)}]
  \input{Tikz/Cell_tower}
  \node at (90,-90) {\color{mcrnred} $\substack{\text{\normalsize {Edge}} \\ \text{\normalsize {Server 2}}}$};
\end{scope}
%%%%%%%%%%%%%%% Edge server Q
\begin{scope}[scale=1.2, shift={(11.1cm,9cm)}]
  \input{Tikz/Cell_tower}
  \node at (90,-90) {\color{mcrnred} $\substack{\text{\normalsize {Edge}} \\ \text{\normalsize {Server $Q$}}}$};
\end{scope}

\draw (575,350) node {\Huge $\cdots\cdots$};

% ==== Downlink (mcrnblue arrows) — static ====
\draw[very thick, color=mcrnbrown, -latex] (215,304) -- (215,210);
\draw[very thick, color=mcrnbrown, -latex] (430,304) -- (430,210);
\draw[very thick, color=mcrnbrown, -latex] (745,304) -- (745,210);
\node[rotate=0] at (265,257) {\color{mcrnbrown} $\substack{\text{\normalsize {Majority-Vote}} \\[0.2em] \text{\normalsize {Signs}}}$};
% ==== Uplink (mcrnred arrows) — static ====
\draw[very thick, dashed, -latex] (204,210) -- (204,304);
\draw[very thick, dashed, -latex] (421,210) -- (421,304);
\draw[very thick, dashed, -latex] (736,210) -- (736,304);
\node[rotate=0] at (168,257) { $\substack{\text{\normalsize {Gradient}} \\[0.2em] \text{\normalsize {Signs}}}$};

%---------- Global Aggregation formula ----------
\draw (610,560) node {\Large{$\mathbf{w}^{(t+1)} = \sum_{q=1}^Q \frac{D_q}{N}\,\mathbf{v}_{q}^{(t,T_E)}$}};
\end{tikzpicture}
}
% \end{figure}
  \vspace{-1.4em}
  \caption{Sign-based implementation for an HFL scenario; devices send gradient signs to their edge servers, the servers broadcast majority-vote results back for local training; after several rounds, the edge servers forward the model parameters to the cloud for global aggregation.}
  \label{fig:HierFedSignSGD}
\end{figure}

For our configuration, depicted in Fig.~\ref{fig:HierFedSignSGD}, we assume that the cloud server manages $Q$ edge servers, each of which is connected to its cluster of devices. For edge server $q$, we denote the device set by $\mathcal{V}^q$. Device $k$ in \(\mathcal{V}^q\) has access to the local dataset \(\mathcal{D}_{qk}\), which it uses to train its local learning model. Let us denote the loss function for a single sample as $\mathcal{L}(\mathbf{w}; \boldsymbol{\xi})$, where $\mathbf{w}$ is the model parameter vector and $\boldsymbol{\xi}$ is the sample vector containing the input and output values. Based on this, the local empirical loss function of device \(k\) associated with edge server \(q\) is defined as \begin{equation} f_{qk}(\mathbf{w}) = \frac{1}{|\mathcal{D}_{qk}|} \sum_{\boldsymbol{\xi} \in \mathcal{D}_{qk}} \mathcal{L}(\mathbf{w}; \boldsymbol{\xi}). \end{equation}
Accordingly, the global average loss function is defined as
\begin{equation} \label{eq:global_loss}
\mathcal{F}(\mathbf{w}) \triangleq \frac{1}{N}\sum_{q = 1}^Q\, \,\sum_{k  \in \mathcal{V}^q} \sum_{\boldsymbol{\xi}  \in \mathcal{D}_{qk}} \mathcal{L}(\mathbf{w}; \boldsymbol{\xi}),
\end{equation}
where $N$ is the total number of data samples in the network.
However, it is beneficial to reformulate $\mathcal{F}(\mathbf{w})$ in a way that mirrors the hierarchical format illustrated earlier. To this end, consider the edge loss functions
\begin{align} \label{eq:Fq(w)}
    \mathcal{F}_q(\mathbf{w}) = \sum_{k  \in \mathcal{V}^q} \frac{|\mathcal{D}_{qk}|}{{D}_{q}}f_{qk}(\mathbf{w}),
\end{align}
where ${D}_{q} = \sum_{k \in \mathcal{V}^q} |\mathcal{D}_{qk}|$. It is evident that \eqref{eq:global_loss} can now be equivalently recast as
\begin{align} \label{eq:simpler global}
    \mathcal{F}(\mathbf{w}) = \sum_{q=1}^Q \frac{D_q}{N}\mathcal{F}_q(\mathbf{w}),
\end{align}
which is a hierarchical representation of the global loss function, as intended.

The ultimate goal is to minimize $\mathcal{F}(\mathbf{w})$:
\begin{align}
    \mathbf{w}^\star = \operatorname*{arg\,min}_{\mathbf{w}} \, \mathcal{F}(\mathbf{w}).
\end{align}
A typical solution for finding $\mathbf{w}^\star$ is to use the SGD approach with step-size $\mu$:
\begin{align}
    &\mathbf{w} \leftarrow \mathbf{w} - \mu \nabla\mathcal{F} (\mathbf{w}) = \mathbf{w} - \mu \sum_{q=1}^Q \frac{D_q}{N}\nabla \mathcal{F}_q(\mathbf{w}) \nonumber \\
    \Leftrightarrow \,\,\, &\mathbf{w} \leftarrow \sum_{q=1}^Q \frac{D_q}{N}\left(\mathbf{w} - \mu\nabla\mathcal{F}_q(\mathbf{w})\right),
    \label{eq:Glob_Aggreg}
\end{align}
where the last representation implies that the edge servers can execute the gradient descent iteration locally and send the parameter models to the cloud server for aggregation, an approach first adopted by the original FL study \cite{mcmahan2017communication}. Furthermore, from \eqref{eq:Fq(w)} we have
\begin{align} \label{eq:edge_Aggreg}
    \nabla\mathcal{F}_q(\mathbf{w}) =  \sum_{k  \in \mathcal{V}^q} \frac{|\mathcal{D}_{qk}|}{{D}_{q}} \nabla f_{qk}(\mathbf{w}),
\end{align}
which indicates that edge servers aggregate the local gradient vectors received from their associated devices, where each device, for the sake of computational efficiency, estimates its gradient using only a small random batch of data samples.

Steps \eqref{eq:Glob_Aggreg} and \eqref{eq:edge_Aggreg} constitute the core of HFL algorithms; however, we encounter specific challenges when executing the second step. In particular, transmitting distinct gradient values from multiple devices to an edge server over a multiple-access channel places considerable strain on communication resources, such as bandwidth, thereby necessitating some form of data quantization. An extreme form of vector quantization preserves only the signs of the entries, discarding all magnitude information. This leads to the SignSGD update rule. In the following section, we introduce our proposed algorithm, which adopts this highly compressed scheme as part of the device--edge training process.

\section{The $\mathtt{HierSignSGD}$ Algorithm} \label{sec:algorithm}

In this section, we introduce the proposed HFL algorithm, \(\mathtt{HierSignSGD}\). We begin with an overview of the algorithm, followed by a detailed convergence analysis.

\subsection{Pseudocode} \label{subsec: pseudo}

The core idea behind \(\mathtt{HierSignSGD}\) is to implement a hierarchical training procedure that operates efficiently under the stringent communication constraints of the device--edge channel. The pseudocode is provided in Algorithm~\ref{alg:hier-SgnSGD}, and a stepwise summary is given below:
\begin{enumerate}
    % \item \mathbf{Global initialization:} 
    % The cloud server initializes the global model \(\mathbf{w}^{(0)}\).

    \item \textbf{Broadcast to edges:}
    At each global round \(t\), the cloud server broadcasts \(\mathbf{w}^{(t)}\) to all edge servers.

    \item \textbf{Initializing device model:}
    Each edge server broadcasts the provided $\mathbf{w}^{(t)}$ to its associated devices, which then set  
    \[
        \mathbf{v}_{q}^{(t,0)} = \mathbf{w}^{(t)}.
    \]
   
    \item \textbf{Local gradient computation at devices:}
    For each local step $\tau$, each device computes a gradient estimate
    \[
        \hat{\mathbf{g}}_{qk}^{(t,\tau)} = \hat{\nabla} f_{qk}\big(\mathbf{v}_{q}^{(t,\tau)}\big).
    \]
    Then, only the element-wise signs \(\operatorname{sgn}\!\big(\hat{\mathbf{g}}_{qk}^{(t,\tau)}\big)\) are sent to the corresponding edge server.

    \item \textbf{Vote-based aggregation at edges:}
    Each edge server aggregates the received signs via a majority-vote
    \[ 
        \mathbf{s}_{q}^{(t,\tau)} = \operatorname{sgn}\!\Big(
            \!\sum_{k \in \mathcal{V}^q} 
            \operatorname{sgn}\!\big(\hat{\mathbf{g}}_{qk}^{(t,\tau)}\big)
        \Big),
    \]
    and transmits the resulting sign vector back to the devices. Subsequently, both the edge server and the devices update their local models using a sign-based descent step:
    \[
        \mathbf{v}_{q}^{(t,\tau+1)} = 
        \mathbf{v}_{q}^{(t,\tau)} - \mu\,\mathbf{s}_{q}^{(t,\tau)}.
    \]

    \item \textbf{Return to cloud for aggregation:}
    After \(T_E\) local steps at the edge, each edge server sends its final local model 
    \(\mathbf{v}_{q}^{(t,T_E)}\) back to the cloud server for the weighted model aggregation:
    \[
        \mathbf{w}^{(t+1)} = 
        \sum_{q=1}^Q \frac{D_q}{N}\,\mathbf{v}_{q}^{(t,T_E)}.
    \]
\end{enumerate}

\begingroup

% Reduce spacing above and below displayed equations
\setlength{\abovedisplayskip}{2pt}
\setlength{\belowdisplayskip}{2pt}
\setlength{\abovedisplayshortskip}{0pt}
\setlength{\belowdisplayshortskip}{2pt}

% Reduce algorithm2e padding
\SetAlgoSkip{}
\SetAlgoInsideSkip{}
\SetAlCapSkip{2pt}

% Optional: smaller font also gives tighter line spacing
\SetAlFnt{\small}

\begin{algorithm}[t]
\caption{$\mathtt{HierSignSGD}$}
\label{alg:hier-SgnSGD}

% \KwIn{Number of global rounds $T_G$, local steps $T_E$, step-size $\mu$, relative edge dataset sizes $D_q/N$;}
% \KwOut{Final global model $\mathbf{w}^{(T_G)}$;}

Initialize global model $\mathbf{w}^{(0)}$\;
\For{$t=0,\ldots,T_G-1$}{
    Cloud broadcasts $\mathbf{w}^{(t)}$ to all edge servers\;
    
    \ForEach{edge server $q=1,\ldots,Q$ in parallel}{
        Broadcast $\mathbf{w}^{(t)}$ to all devices $k\in\mathcal{V}^q$\;
        
        Set
        \[
            \mathbf{v}_q^{(t,0)}=\mathbf{w}^{(t)};
        \]
        
        \For{$\tau=0,\ldots,T_E-1$}{
            Each device $k\in\mathcal{V}^q$ computes
            \[
                \hat{\mathbf{g}}_{qk}^{(t,\tau)}
                =
                \hat{\nabla} f_{qk}
                \big(
                    \mathbf{v}_q^{(t,\tau)}
                \big),
            \]
            and sends $\operatorname{sgn}\!\big(\hat{\mathbf{g}}_{qk}^{(t,\tau)}\big)$ to the edge server;
            
            Edge server computes the majority-vote
            \[
                \mathbf{s}_q^{(t,\tau)}
                =
                \operatorname{sgn}
                \Big(\!
                    \sum_{k\in\mathcal{V}^q}
                    \operatorname{sgn}
                    \big(
                        \hat{\mathbf{g}}_{qk}^{(t,\tau)}
                    \big)
                \Big),
            \]
            
            and sends  $\mathbf{s}_q^{(t,\tau)}$ back to the devices;
            
            Jointly update
            \[
                \mathbf{v}_q^{(t,\tau+1)}
                =
                \mathbf{v}_q^{(t,\tau)}
                -
                \mu \mathbf{s}_q^{(t,\tau)};
            \]
        }
        
        Send $\mathbf{v}_q^{(t,T_E)}$ to the cloud server\;
    }
    
    Cloud aggregates
    \[
        \mathbf{w}^{(t+1)}
        =
        \sum_{q=1}^Q
        \frac{D_q}{N}
        \mathbf{v}_q^{(t,T_E)} .
    \]
}
\end{algorithm}

\endgroup

We note that the use of sign-based updates to train edge models across multiple clusters directly influences the convergence behavior of the algorithm. A key objective of the following analysis is to characterize how this extreme form of compression, repeatedly applied over multiple device--edge communication rounds, impacts the convergence of \(\mathtt{HierSignSGD}\).

\subsection{Convergence Analysis}
We now aim to analyze the convergence of the proposed \(\mathtt{HierSignSGD}\) algorithm. The main objective is to characterize the expected asymptotic behavior of the iterates $\mathbf{w}^{(t)}$ produced by the algorithm. We do this by providing an upper bound for
\begin{align} \label{eq:sigma_converge}
     \mathbb{E}\Big\{ \frac{1}{T_G}\!\sum_{t=0}^{T_G-1} \big\| \nabla  \mathcal{F}({\mathbf{w}}^{(t)}) \big\|_1\Big\}.
\end{align}
As can be seen, we use the $\ell_1$-norm in \eqref{eq:sigma_converge}, while its dual norm, 
$\ell_\infty$, appears in intermediate steps of the analysis later. This choice is dictated by the geometry 
of sign-based updates rather than by an arbitrary norm selection. In particular, 
as we shall see, the descent term induced by a sign direction naturally 
corresponds to $\|\cdot\|_1$ of the gradient.

% Therefore, the $\ell_1$-norm 
% provides a natural stationarity measure for SignSGD-type methods.
% This choice is mainly dictated by the geometry of sign-based updates.
% Hence, we select norms solely for analytical convenience, without loss of generality in the results.
% Because we mostly deal with inequalities in our analysis, the choice of norm is inconsequential, as all norms on a finite-dimensional space are equivalent~\cite{horn2013matrix}. 

Our analysis relies on the following standard assumptions commonly adopted in the FL literature:
\smallskip
\begin{itemize}
    \item[$\blacktriangleright$]  \textbf{A1) Lower bounded objective}:
    For all $\mathbf{w}\in \mathbb{R}^d$, we have 
    \[
    \mathcal{F}(\mathbf{w}) \geq \mathcal{F}^\star,
    \]
    where \(\mathcal{F}^\star\) is a lower bound on the objective value.
    \smallskip
    \item[$\blacktriangleright$]  \textbf{A2) Smoothness}: Each loss function $\mathcal{F}_q: \mathbb{R}^d \rightarrow \mathbb{R}$ is $L$-smooth with respect to $\|\cdot\|$:
    \[
    \left\|\nabla \mathcal{F}_q(\mathbf{v})-\nabla \mathcal{F}_q(\mathbf{w})\right\|_{*}\, \le \,L\,\|\mathbf{v}-\mathbf{w}\|,
    \quad \forall\,\mathbf{w},\mathbf{v}\in \mathbb{R}^d,
    \]
    which implies
    \[
    \mathcal{F}_q(\mathbf{v})\ \le\ \mathcal{F}_q(\mathbf{w})\,+\,\left\langle \nabla \mathcal{F}_q(\mathbf{w}),\,\mathbf{v}-\mathbf{w}\right\rangle\,+\,\frac{L}{2}\,\|\mathbf{v}-\mathbf{w}\|^{2}.
    \]
    Consequently, the global loss function $\mathcal{F}$ in~\eqref{eq:simpler global} will also possess this property. 
    \smallskip
    \item[$\blacktriangleright$]  \textbf{A3) Bounded variance}: Each stochastic gradient obtained from a random sample is an unbiased estimator of the full-batch gradient, with its coordinates having bounded variance:
    \[
    \mathbb{E}\{\hat{\mathbf{g}}(\mathbf{w})\} = \mathbf{g}(\mathbf{w}),
    \quad
    \mathbb{E}\Big\{\!\big([\hat{\mathbf{g}}(\mathbf{w})]_i - [\mathbf{g}(\mathbf{w})]_i\big)^2\Big\} \le \sigma^2.
    \]
    It follows from this assumption that the mini-batch gradient estimate is also unbiased, with variance bound reduced to $\sigma^2/B$, where $B$ denotes the batch-size. By employing the identity $\mathbb{E}\{X\} \leq \sqrt{\mathbb{E}\{X^2\}}$, we deduce
    \[
    \mathbb{E}\Big\{\big|[\hat{\mathbf{g}}(\mathbf{w})]_i - [\mathbf{g}(\mathbf{w})]_i\big| \Big\} \le \frac{\sigma}{\sqrt{B}}.
    \]
    %\smallskip
    \item[$\blacktriangleright$] \textbf{A4) Inter-cluster gradient dissimilarity}: Throughout this paper, inter-cluster data heterogeneity refers to the statistical mismatch among the data distributions of device clusters associated with different edge servers. We quantify its effect through the following edge-level gradient dissimilarity measure: \begin{align*} \zeta = \sup_{\mathbf{w}\in\mathbb{R}^d} \sum_{q=1}^Q \frac{D_q}{N} \left\| \nabla \mathcal{F}_q(\mathbf{w}) - \nabla \mathcal{F}(\mathbf{w}) \right\|_1. \end{align*} The quantity \(\zeta\) measures the weighted mismatch between the edge-level gradients and the global gradient, and is commonly used as a measure of statistical heterogeneity or non-IIDness (see, e.g., \cite{yu2019linear,haddadpour2019convergence}).
\end{itemize}

%%%%%%%%%%%%%%%%%%%%%%%%%%%%%%%%%%%%%%%%%%%%%%%%%%%%%%%%%%%%%%%%%%%%
Given assumptions A1--A4, we proceed to analyze the convergence properties of the \(\mathtt{HierSignSGD}\) algorithm. We first analyze a simplified variant of the algorithm that omits the majority-vote mechanism. We then extend the analysis to demonstrate that the same error bound holds when majority voting is employed at the edge servers.

\begin{theorem} \label{thm:1st}
Consider running Algorithm~\ref{alg:hier-SgnSGD} with single-device clusters for $T_G$ global rounds and $T_E$ local steps, using step-size $\mu$ and batch-size $B$. 
Under assumptions A1--A4, the following performance bound holds:
\begin{align}\label{eq:bound}
\frac{1}{T_G}\!\!\sum_{t=0}^{T_G-1} 
\mathbb{E}\!\left\{\big\|\nabla \mathcal{F}(\mathbf{w}^{(t)})\big\|_{1}\right\}
\;\le\;
\frac{\mathcal{F}(\mathbf{w}^{(0)})-\mathcal{F}^\star}{\mu T_GT_E}
\;+\;C,
\end{align}
where
\begin{align} \label{eq:C}
 C \,=\, 2\zeta \,+\, \frac{2\sigma d}{\sqrt{B}} \,+\, \big(\frac{3T_E}{2}-1\big)L\mu.   
\end{align}

\end{theorem}

\begin{proof}
    See Appendix~\ref{app:thm1}
\end{proof}

\begin{remark} \label{Rem:zeta_effect}
The first term on the right-hand side of~\eqref{eq:bound}
vanishes as the number of global rounds \(T_G\) increases. Moreover,
by properly tuning the step-size \(\mu\) and the batch-size \(B\), the stochastic
gradient error and the local-drift terms appearing in \(C\) can be reduced.
However, the term \(2\zeta\) cannot be controlled by such algorithmic
hyperparameters, and therefore constitutes an irreducible bias floor in the bound. This indicates that the original sign-based HFL procedure cannot directly compensate for the drift of edge models toward their local objectives. By contrast, in full-precision SGD-based FL methods, heterogeneity typically manifests through local model drift whose effect can be attenuated by decreasing the step-size \cite{karimireddy2020scaffold}.
\end{remark}
%\textcolor{}{

\section{Correcting the Edge Gradient Drift} \label{sec:DC-algorithm}

As highlighted in Remark~\ref{Rem:zeta_effect}, the inter-cluster heterogeneity term \(2\zeta\) in \eqref{eq:C} is
problematic because it does not vanish by increasing \(T_G\), decreasing
\(\mu\), or increasing \(B\). 
Furthermore, the definition of $\zeta$ shows that even IID data partitioning across edge servers does not necessarily imply $\zeta = 0$. For finite empirical datasets, the edge-level objectives can still differ from the global objective due to finite-sample effects, even when all samples are drawn from the same underlying distribution. The value of $\zeta$ vanishes only in the idealized population-risk limit, or asymptotically as the number of samples per edge becomes sufficiently large. Therefore, from the perspective of finite-sample convergence guarantees, a small heterogeneity-induced bias may persist even under IID sampling. This observation motivates the need for an algorithmic correction mechanism to ensure convergence.

In full-precision hierarchical federated optimization, local drift among edge servers can be mitigated by applying correction terms directly to the local stochastic gradient vectors~\cite{fang2024hierarchical}, yielding an update direction of the form \(\hat{\mathbf{g}}_{qk}^{(t,\tau)}+\boldsymbol{\delta}_q^{(t)}\), where \(\boldsymbol{\delta}_q^{(t)}\) is a correction vector. However, the sign-based setting considered here poses an additional challenge. During local training, devices transmit only coordinate-wise signs, so the edge server does not observe the stochastic gradients themselves. Hence, this full-vector correction cannot be applied after sign transmission, since the magnitude information in \(\hat{\mathbf{g}}_{qk}^{(t,\tau)}\) has already been discarded. Instead, the correction must be incorporated at the device before the sign operation is taken. This modifies the sign decision itself while preserving binary device--edge communication. Whether such a pre-sign correction remains effective after one-bit compression is nontrivial and constitutes a key focus of our analysis.

To this end, we introduce edge-level and global gradient anchors. For edge
server \(q\), define
\begin{align}
    \mathbf{c}_q^{(t)}
    &\triangleq
    \nabla \mathcal{F}_q(\mathbf{w}^{(t)}),
    \label{eq:edge_anchor}
    \\
    \mathbf{c}^{(t)}
    &\triangleq
    \sum_{q=1}^Q \frac{D_q}{N}\mathbf{c}_q^{(t)}
    =
    \nabla \mathcal{F}(\mathbf{w}^{(t)}).
    \label{eq:global_anchor}
\end{align}
The vector \(\mathbf{c}_q^{(t)}\) represents the gradient anchor induced by the data distribution of edge \(q\), whereas \(\mathbf{c}^{(t)}\) represents the corresponding global gradient anchor. Hence, the drift-correction vector
for edge \(q\) is given by
\begin{align}
    \boldsymbol{\delta}_q^{(t)}
    =
    \mathbf{c}^{(t)}-\mathbf{c}_q^{(t)}.
    \label{eq:correction_vector}
\end{align}
Ideally, using \(\boldsymbol{\delta}_q^{(t)}\) would align the average corrected
edge-level direction with the global direction at \(\mathbf{w}^{(t)}\). Indeed,
at the beginning of a global round, where
\(\mathbf{v}_q^{(t,0)}=\mathbf{w}^{(t)}\), we have
\begin{align*}
    \nabla \mathcal{F}_q(\mathbf{w}^{(t)})
    +
    \mathbf{c}^{(t)}
    -
    \mathbf{c}_q^{(t)}
    =
    \nabla \mathcal{F}(\mathbf{w}^{(t)}).
\end{align*}
Thus, the edge-specific gradient bias is canceled at the start of the round.

However, in practice, one should use a damped correction by applying
\begin{align}  
\hat{\mathbf{g}}_{qk}^{(t,\tau)}+\rho\boldsymbol{\delta}_q^{(t)},
\end{align} 
 where
\(\rho\in(0,1]\) is a tunable correction-strength parameter. The case
\(\rho=1\) corresponds to the full correction described  earlier, while using smaller values may improve stability when $T_E$ is large.

A direct implementation of~\eqref{eq:edge_anchor}--\eqref{eq:global_anchor}
would require the cloud to first broadcast \(\mathbf{w}^{(t)}\), receive
\(\mathbf{c}_q^{(t)}\) from all edges, compute \(\mathbf{c}^{(t)}\), and then
broadcast \(\mathbf{c}^{(t)}\) back before local training starts. To avoid this additional pretraining synchronization, we use a pipelined version of the
correction. Specifically, we initialize
\begin{align}
    \mathbf{c}^{(-1)}=\mathbf{0},
    \qquad
    \mathbf{c}_q^{(-1)}=\mathbf{0},
    \quad q=1,\ldots,Q,
\end{align}
and at global round \(t\), the devices under edge \(q\) use the one-round stale
correction $\mathbf{c}^{(t-1)}-\mathbf{c}_q^{(t-1)}$.
Accordingly, device \(k\in\mathcal{V}^q\) computes the corrected sign
\begin{align}
    \widetilde{\mathbf{s}}_{qk}^{(t,\tau)}
    =
    \operatorname{sgn}\!
    \big(
        \hat{\mathbf{g}}_{qk}^{(t,\tau)}
        +
        \rho\mathbf{c}^{(t-1)}
        -
        \rho\mathbf{c}_q^{(t-1)}
    \big),
    \label{eq:corrected_device_sign}
\end{align}
and sends only \(\widetilde{\mathbf{s}}_{qk}^{(t,\tau)}\) to the edge server.
The edge server then performs majority-vote aggregation over the corrected
signs,
\begin{align}
    \widetilde{\mathbf{s}}_q^{(t,\tau)}
                =
                \operatorname{sgn}\!
                \Big(
                    \!\sum_{k\in\mathcal{V}^q}
                    \widetilde{\mathbf{s}}_{qk}^{(t,\tau)}
                \Big),
    \label{eq:corrected_majority_vote}
\end{align}
and updates its local model according to
\begin{align*}
    \mathbf{v}_q^{(t,\tau+1)}
    =
    \mathbf{v}_q^{(t,\tau)}
    -
    \mu
    \widetilde{\mathbf{s}}_q^{(t,\tau)}.
    % \label{eq:corrected_update}
\end{align*}

During the same global round, the devices also compute anchor gradients $\nabla f_{qk}(\mathbf{w}^{(t)})$ at the
broadcast model \(\mathbf{w}^{(t)}\). The edge server aggregates these anchor
gradients as
\begin{align}
    \mathbf{c}_q^{(t)}
    =
    \sum_{k\in\mathcal{V}^q}
    \frac{|\mathcal{D}_{qk}|}{D_q}
    \nabla f_{qk}(\mathbf{w}^{(t)}).
    \label{eq:edge_anchor_update}
\end{align}
After completing the \(T_E\) local corrected sign steps, the edge server 
sends both \(\mathbf{v}_q^{(t,T_E)}\) and \(\mathbf{c}_q^{(t)}\) to the cloud for aggregation
\begin{align}
    \mathbf{w}^{(t+1)}
    =
    \sum_{q=1}^Q
    \frac{D_q}{N}
    \mathbf{v}_q^{(t,T_E)},
    \quad
    \mathbf{c}^{(t)}
    =
    \sum_{q=1}^Q
    \frac{D_q}{N}
    \mathbf{c}_q^{(t)}.
    \label{eq:pipelined_cloud_updates}
\end{align}
The pair \((\mathbf{w}^{(t+1)},\mathbf{c}^{(t)})\) is then broadcast at the
beginning of the next global round. The complete procedure is summarized in Algorithm~\ref{alg:dc-hiersignsgd}.
\begin{remark} \label{Rem:price_paid}
    The additional transmission of anchor gradients in~\eqref{eq:edge_anchor_update}
is the price paid for mitigating heterogeneity-induced edge drift. Importantly, this
cost is incurred only once per global round, while the repeated device--edge
communication during the \(T_E\) local training steps remains binary. Hence,
the proposed correction trades one low-frequency full-gradient synchronization
for a substantial reduction in the heterogeneity-induced drift that appears as
the \(2\zeta\) term in the original bound.
\end{remark}

\begingroup

% Reduce spacing above and below displayed equations
\setlength{\abovedisplayskip}{2pt}
\setlength{\belowdisplayskip}{2pt}
\setlength{\abovedisplayshortskip}{0pt}
\setlength{\belowdisplayshortskip}{2pt}

% Reduce algorithm2e padding
\SetAlgoSkip{}
\SetAlgoInsideSkip{}
\SetAlCapSkip{2pt}

% Optional: smaller font also gives tighter line spacing
\SetAlFnt{\small}

\begin{algorithm}[t]
\caption{$\mathtt{DC\text{-}HierSignSGD}$}
\label{alg:dc-hiersignsgd}

% \KwIn{Global rounds $T_G$, local steps $T_E$, step-size $\mu$, weights $D_q/N$ and $|\mathcal{D}_{qk}|/D_q$, and parameter $\rho$;}
% \KwOut{Global model $\mathbf{w}^{(T_G)}$;}

Initialize $\mathbf{w}^{(0)}$ and set $\mathbf{c}^{(-1)}=\mathbf{0}$ and $\mathbf{c}_q^{(-1)}=\mathbf{0}, \ \forall q$\;

\For{$t=0,\ldots,T_G-1$}{
    Cloud broadcasts $\big(\mathbf{w}^{(t)},\mathbf{c}^{(t-1)}\big)$ to all edge servers\;
    
    \ForEach{edge server $q=1,\ldots,Q$ \emph{in parallel}}{
        Set $\mathbf{v}_q^{(t,0)}=\mathbf{w}^{(t)}$ and
        $\boldsymbol{\delta}_q^{(t-1)}=\mathbf{c}^{(t-1)}-\mathbf{c}_q^{(t-1)}$\;
        
        Broadcast $\mathbf{w}^{(t)}$ and $\boldsymbol{\delta}_q^{(t-1)}$ to devices $k\in\mathcal{V}^q$\;
        
        Each device computes
        $\nabla f_{qk}(\mathbf{w}^{(t)})$
        and sends it to edge server\;
        
        Edge server computes and holds
        \[
            \mathbf{c}_q^{(t)}
            =
            \sum_{k\in\mathcal{V}^q}
            \frac{|\mathcal{D}_{qk}|}{D_q}\nabla f_{qk}(\mathbf{w}^{(t)});
        \]
        
        \For{$\tau=0,\ldots,T_E-1$}{
            Each device $k\in\mathcal{V}^q$ computes
            \[                \widetilde{\mathbf{s}}_{qk}^{(t,\tau)}
                =
                \operatorname{sgn}\!
                \Big(
                    \hat{\nabla} f_{qk}
                    \big(
                        \mathbf{v}_q^{(t,\tau)}
                    \big)
                    +
                    \rho\boldsymbol{\delta}_q^{(t-1)}
                \Big),
            \]
            and sends $\widetilde{\mathbf{s}}_{qk}^{(t,\tau)}$ to edge server $q$\;
            
            Edge server computes the  majority vote
            \[
                \widetilde{\mathbf{s}}_q^{(t,\tau)}
                =
                \operatorname{sgn}\!
                \Big(
                    \!\sum_{k\in\mathcal{V}^q}
                    \widetilde{\mathbf{s}}_{qk}^{(t,\tau)}
                \Big);
            \]
            
            Update
            \[
                \mathbf{v}_q^{(t,\tau+1)}
                =
                \mathbf{v}_q^{(t,\tau)}
                -
                \mu
                \widetilde{\mathbf{s}}_q^{(t,\tau)};
            \]
        }
        
        Send $\big(\mathbf{v}_q^{(t,T_E)},\mathbf{c}_q^{(t)}\big)$ to the cloud\;
    }
    
    Cloud updates
    \[
        \mathbf{w}^{(t+1)}
        =
        \sum_{q=1}^Q
        \frac{D_q}{N}
        \mathbf{v}_q^{(t,T_E)},
        \quad
        \mathbf{c}^{(t)}
        =
        \sum_{q=1}^Q
        \frac{D_q}{N}
        \mathbf{c}_q^{(t)} .
    \]
}
\end{algorithm}

\endgroup

Continuing with the single-device clusters, Algorithm~\ref{alg:dc-hiersignsgd} results in the following performance bound.
\begin{theorem} \label{thm:2nd}
With gradient correction applied at the devices during each local step, and under the assumptions of Theorem~\ref{thm:1st}, Algorithm~\ref{alg:dc-hiersignsgd} attains the following error bound:
\begin{align}\label{eq:new_bound}
\frac{1}{T_G}\!\sum_{t=0}^{T_G-1} 
\mathbb{E}\!\left\{\big\|\nabla \mathcal{F}(\mathbf{w}^{(t)})\big\|_{1}\right\}
\le
\frac{\mathcal{F}(\mathbf{w}^{(0)})-\mathcal{F}^\star}{\mu T_GT_E}
+{C}_{dc},
\end{align}
where
\begin{align}\label{eq:C_tilde}
{C}_{dc} \,=\, 2(1-\rho)\zeta+\frac{2\sigma d}{\sqrt{B}} \,+\, \Big(\frac{(3+8\rho)T_E}{2}-1\Big)L\mu.
\end{align}
\end{theorem}

\begin{proof}
    See Appendix~\ref{app:thm2}.
\end{proof}

Let us now comment on the above bound.

\smallskip

\begin{itemize}

\item \textbf{Mitigation of the \(\zeta\) term.} 
The most noticeable consequence of applying gradient correction at the devices is the reduction of the gradient dissimilarity effect in the convergence bound. The correction counteracts the direct influence of inter-cluster heterogeneity, which otherwise appears as a non-vanishing bias term. When the full correction with $\rho=1$ is used, the bias term is canceled; when a damped correction with \(0<\rho<1\) is used, the heterogeneity-induced drift is only partially compensated, but the update may become more stable for larger $T_E$. Larger correction strengths can accelerate early progress by more aggressively compensating inter-cluster drift, but may induce oscillations in later rounds when gradients become small and the correction begins to dominate the sign decision. Thus, $\rho$ controls a stability--correction tradeoff that is closely coupled with the number of local steps $T_E$, as will be demonstrated in our simulations. 

\smallskip

\item \textbf{Effect of $T_E$.}
The number of local steps $T_E$ plays a central role in the behavior of the corrected sign-based method. As evident from \eqref{eq:new_bound}, increasing \(T_E\) reduces the frequency of cloud aggregation and therefore improves communication efficiency, but it also allows the edge models to drift farther from the point at which the correction term was computed. As a result, the correction may become less representative of the current local gradients when \(T_E\) is large, which can lead to  oscillatory and unstable behavior in sign-based updates. In practice, the correction strength~\(\rho\) and the number of local steps \(T_E\) should therefore be tuned jointly.

\smallskip

    \item \textbf{Selecting hyperparameters.} The convergence bound also clarifies the roles of the step-size and batch-size. The optimization term improves with a larger effective step-size, whereas the drift and smoothness-related terms grow with the step-size and the number of local edge steps. Thus, the step-size controls a fundamental tradeoff between descent speed and accumulated local drift. The batch-size affects the stochastic-gradient error; larger batches reduce the variance of the gradient estimates and improve sign reliability, but require more local computation. If constant values are used throughout training, they should therefore be chosen in a horizon-dependent manner; for example, setting $\mu = 1/\sqrt{T_G}$ and $B = T_G$ balances the optimization, drift, and stochastic-gradient error terms in the convergence bound. The following corollary highlights this.
\end{itemize}

\begin{corollary} \label{cor:2}
Choosing \(\mu = 1/\sqrt{T_G}\) and \(B = T_G\), we obtain a worst-case sublinear convergence rate of $\mathcal{O}(1/\sqrt{T_G})$ for $\mathtt{DC\text{-}HierSignSGD}$ with $\rho=1$:
\begin{align} 
\frac{1}{T_G}\!\!\sum_{t=0}^{T_G-1} 
\!\mathbb{E}\!\left\{\big\|\nabla \mathcal{F}(\mathbf{w}^{(t)})\big\|_{1}\right\}
\le\frac{1}{\sqrt{T_G}}\Big(
\frac{\mathcal{F}(\mathbf{w}^{(0)})-\mathcal{F}^\star}{ T_E }
+ \widetilde{C}_{dc}\Big),
\end{align}
where
\begin{align*}
    \widetilde{C}_{dc}\,=\,2\sigma d + \big(\frac{11T_E}{2}-1\big)L.
\end{align*}
\end{corollary}
From Corollary~\ref{cor:2}, it follows that as $T_G \rightarrow \infty$, we get
\begin{equation}
    \min_{\vphantom{\int_0^A}\,0 \le t \le T_G - 1}
    \mathbb{E}\!\left\{ \big\|\nabla \mathcal{F}(\mathbf{w}^{(t)})\big\|_1 \right\}
    \;\longrightarrow\; 0,
\end{equation}
indicating that the algorithm produces a sequence whose \emph{best iterate} converges on average to a stationary point.

\smallskip

% Extending the obtained results to the majority-vote setting with multi-device clusters, we have the following theorem.

We now extend the result to the majority-vote setting with multiple devices per edge cluster. Intuitively, when devices within each edge cluster have IID data distributions and their sign errors are conditionally independent, majority voting does not reduce sign reliability if individual signs are more likely to be correct than incorrect. Hence, the single-device sign-error bound also applies to the aggregated edge-level sign, leading to the following theorem.

\begin{theorem} \label{thm:3rd} Under the assumptions of Theorems~\ref{thm:1st} and~\ref{thm:2nd}, the bounds in~\eqref{eq:bound} and~\eqref{eq:new_bound} continue to hold when each edge server aggregates the signs of its associated devices through majority voting, provided that data are IID within each edge cluster.
\end{theorem}

\begin{proof}
 See Appendix~\ref{App_A}.
\end{proof}
\begin{remark} \label{Rem:future_work}
In this work, the non-IID setting is designed to emphasize heterogeneity across edge servers. This choice reflects a common hierarchical edge-learning scenario in which devices associated with the same edge server are geographically or contextually close and can therefore be expected to have relatively similar data distributions, while different edge servers may serve distinct regions, user populations, or sensing environments. Accordingly, our model treats inter-cluster heterogeneity as the primary source of statistical mismatch. A more general analysis of sign-based HFL with both intra-cluster and inter-cluster heterogeneity would require an additional device-level gradient dissimilarity measure within each cluster, which we leave for future work.
\end{remark}

\section{Simulations} \label{sec:simulations}
To evaluate the performance of the proposed sign-based HFL methods, we conduct experiments on EMNIST-Digits, Fashion-MNIST, and CIFAR-10 datasets. EMNIST-Digits serves as a standard benchmark for large-scale distributed learning, while Fashion-MNIST and CIFAR-10 allow us to examine the proposed methods under progressively more challenging image-classification settings.
However, the remainder of the simulations are conducted only on the EMNIST-Digits, as its simple learning model allows us to better isolate and understand the effects of the key parameters in the sign-based update rule without confounding them with the complexity of a deeper model.
%\vspace{-0.3cm}
%\input{Figs/histo_nonIID}

\subsection{Setup}

\begin{figure*}[t]
    \centering
    \newcommand{\figspace}{\hspace{0.02\textwidth}}
    
    \subfloat[EMNIST-Digits, IID]{%
        \begin{minipage}{\SCALE\textwidth}
            \centering
            \resizebox{\linewidth}{!}{% This file was created with tikzplotlib v0.10.1.post13.
\begin{tikzpicture}[scale=0.47]

\definecolor{crimson2143940}{RGB}{214,39,40}
\definecolor{darkgray176}{RGB}{176,176,176}
\definecolor{darkorange25512714}{RGB}{255,127,14}
\definecolor{forestgreen4416044}{RGB}{44,160,44}
\definecolor{lightgray204}{RGB}{204,204,204}
\definecolor{steelblue31119180}{RGB}{31,119,180}

\begin{axis}[
legend cell align={left},
legend style={
  fill opacity=0.8,
  draw opacity=1,
  text opacity=1,
  at={(1,0.0)},
  anchor=south east,
  draw=lightgray204
},
tick align=outside,
tick pos=left,
x grid style={darkgray176},
xlabel={Global Round},
xmajorgrids,
xmin=-0.45, xmax=31.45,
xtick style={color=black},
y grid style={darkgray176},
ylabel={Test Accuracy (\%)},
ymajorgrids,
ymin=76.993875, ymax=96.708625,
ytick style={color=black},
ytick distance = 2
]
\addplot [semithick, steelblue31119180, mark=square*, mark size=1, mark options={solid}]
table {%
1 80.9625
2 85.735
3 87.865
4 89.1375
5 90.0575
6 90.6575
7 91.09
8 91.56
9 91.815
10 92.0775
11 92.2975
12 92.5225
13 92.6625
14 92.845
15 92.7625
16 93.07
17 93.18
18 93.33
19 93.4275
20 93.495
21 93.6575
22 93.715
23 93.775
24 93.8925
25 93.9425
26 94.0175
27 94.0075
28 94.16
29 94.12
30 94.2925
};
\addlegendentry{$\mathtt{HierSGD}$}
\addplot [semithick, darkorange25512714, mark=diamond*, mark size=1, mark options={solid}]
table {%
1 77.65
2 85.15
3 87.325
4 88.8875
5 89.4275
6 90.4
7 90.775
8 91.2225
9 91.6675
10 91.89
11 92.1375
12 92.3425
13 92.21
14 92.5375
15 92.54
16 92.8
17 92.9075
18 93.035
19 93.0975
20 93.2825
21 93.33
22 93.48
23 93.545
24 93.6975
25 93.7375
26 93.8
27 93.925
28 93.9975
29 93.6075
30 94.195
};
\addlegendentry{$\mathtt{Hier\text{-}Local\text{-}QSGD}$}
\addplot [semithick, forestgreen4416044, mark=*, mark size=1, mark options={solid}]
table {%
1 82.565
2 87.1075
3 88.9525
4 90.1625
5 91.34
6 91.6375
7 92.5175
8 92.8775
9 93.2625
10 93.5925
11 93.7925
12 93.9525
13 94.2575
14 94.37
15 94.4225
16 94.495
17 94.7025
18 94.83
19 94.9025
20 94.9725
21 95.055
22 95.165
23 95.2525
24 95.2675
25 95.4525
26 95.465
27 95.575
28 95.6
29 95.6325
30 95.8
};
\addlegendentry{$\mathtt{HierSignSGD}$}
\addplot [semithick, crimson2143940, mark=triangle*, mark size=1, mark options={solid}]
table {%
1 83.0875
2 86.8275
3 89.11
4 89.95
5 91.25
6 92.0075
7 92.5675
8 92.69
9 93.365
10 93.5075
11 93.7975
12 94.095
13 94.2275
14 94.3225
15 94.495
16 94.5825
17 94.71
18 94.8325
19 94.8725
20 95.0875
21 95.19
22 95.22
23 95.3025
24 95.355
25 95.505
26 95.4825
27 95.605
28 95.69
29 95.6925
30 95.8125
};
\addlegendentry{$\mathtt{DC\text{-}HierSignSGD}$}
\end{axis}

\end{tikzpicture}}%
        \end{minipage}
        \label{fig:SignvsGD_IID}
    }
    \figspace
    \subfloat[EMNIST-Digits, non-IID ($\alpha = 0.1$)]{%
        \begin{minipage}{\SCALE\textwidth}
            \centering
            \resizebox{\linewidth}{!}{% This file was created with tikzplotlib v0.10.1.post13.
\begin{tikzpicture}[scale=0.47]

\definecolor{crimson2143940}{RGB}{214,39,40}
\definecolor{darkgray176}{RGB}{176,176,176}
\definecolor{darkorange25512714}{RGB}{255,127,14}
\definecolor{forestgreen4416044}{RGB}{44,160,44}
\definecolor{lightgray204}{RGB}{204,204,204}
\definecolor{steelblue31119180}{RGB}{31,119,180}

\begin{axis}[
legend cell align={left},
legend style={
  fill opacity=0.8,
  draw opacity=1,
  text opacity=1,
  at={(1,0.0)},
  anchor=south east,
  draw=lightgray204
},
tick align=outside,
tick pos=left,
x grid style={darkgray176},
xlabel={Global Round},
xmajorgrids,
xmin=-0.45, xmax=31.45,
xtick style={color=black},
y grid style={darkgray176},
ylabel={Test Accuracy (\%)},
ymajorgrids,
ymin=13.378625, ymax=98.218875,
ytick style={color=black},
ytick distance = 8
]
\addplot [semithick, steelblue31119180, mark=square*, mark size=1, mark options={solid}]
table {%
1 28.9925
2 52.4325
3 61.285
4 67.5375
5 74.115
6 77.1175
7 80.5475
8 82.17
9 83.75
10 84.8425
11 85.765
12 86.5075
13 87.1675
14 87.68
15 87.93
16 88.335
17 88.9425
18 89.005
19 89.4125
20 89.675
21 89.885
22 90.0025
23 90.2525
24 90.35
25 90.585
26 90.715
27 90.8575
28 91.0175
29 91.0625
30 91.2525
};
\addlegendentry{$\mathtt{HierSGD}$}
\addplot [semithick, darkorange25512714, mark=diamond*, mark size=1, mark options={solid}]
table {%
1 26.125
2 46.44
3 58.37
4 68.95
5 72.0925
6 76.8125
7 80.045
8 80.58
9 81.6675
10 83.555
11 84.6775
12 84.4025
13 85.7825
14 86.6525
15 86.8575
16 87.175
17 87.475
18 87.785
19 88.295
20 88.5425
21 88.8275
22 88.9925
23 89.3575
24 89.3075
25 89.5425
26 89.74
27 90.055
28 90.0625
29 90.3525
30 90.1175
};
\addlegendentry{$\mathtt{Hier\text{-}Local\text{-}QSGD}$}
\addplot [semithick, forestgreen4416044, mark=*, mark size=1, mark options={solid}]
table {%
1 42.52
2 46.4825
3 47.8925
4 56.5725
5 60.3325
6 60.97
7 62.75
8 67.29
9 68.855
10 71.7225
11 74.03
12 76.2775
13 76.8775
14 79.21
15 79.3175
16 80.325
17 80.76
18 81.4525
19 82.5825
20 82.27
21 82.2275
22 82.79
23 83.6775
24 83.6025
25 83.38
26 83.775
27 84.725
28 85.565
29 85.935
30 86.3375
};
\addlegendentry{$\mathtt{HierSignSGD}$}
\addplot [semithick, crimson2143940, mark=triangle*, mark size=1, mark options={solid}]
table {%
1 40.5775
2 54.095
3 79.93
4 84.6475
5 87.485
6 89.105
7 89.7425
8 89.505
9 90.7875
10 91.5225
11 91.7175
12 91.6625
13 92.055
14 92.4325
15 92.7025
16 92.9275
17 92.9575
18 93.0225
19 93.1125
20 93.1975
21 93.67
22 93.7
23 93.8175
24 93.89
25 93.7875
26 93.985
27 94.2075
28 93.975
29 94.105
30 94.3625
};
\addlegendentry{$\mathtt{DC\text{-}HierSignSGD}$}
\end{axis}

\end{tikzpicture}}%
        \end{minipage}
        \label{fig:SignvsGD_nonIID}
    }
    \figspace
    \subfloat[Fashion-MNIST, IID]{%
        \begin{minipage}{\SCALE\textwidth}
            \centering
            \resizebox{\linewidth}{!}{% This file was created with tikzplotlib v0.10.1.post13.
\begin{tikzpicture}[scale=0.47]

\definecolor{crimson2143940}{RGB}{214,39,40}
\definecolor{darkgrey176}{RGB}{176,176,176}
\definecolor{darkorange25512714}{RGB}{255,127,14}
\definecolor{forestgreen4416044}{RGB}{44,160,44}
\definecolor{lightgrey204}{RGB}{204,204,204}
\definecolor{steelblue31119180}{RGB}{31,119,180}

\begin{axis}[
legend cell align={left},
legend style={
  fill opacity=0.8,
  draw opacity=1,
  text opacity=1,
  at={(1,0.0)},
  anchor=south east,
  draw=lightgrey204
},
tick align=outside,
tick pos=left,
x grid style={darkgrey176},
xlabel={Global Round},
xmajorgrids,
xmin=-0.45, xmax=31.45,
xtick style={color=black},
y grid style={darkgrey176},
ylabel={Test Accuracy (\%)},
ymajorgrids,
ymin=40.3755, ymax=86.8945,
ytick style={color=black},
ytick distance = 5
]
\addplot [semithick, steelblue31119180, mark=square*, mark size=1, mark options={solid}]
table {%
1 42.49
2 65.57
3 59.48
4 64.58
5 71.46
6 71.02
7 73.33
8 76.5
9 76.07
10 76.86
11 76.45
12 77.86
13 79.31
14 78.97
15 78.35
16 79.93
17 80.66
18 80.15
19 81.76
20 80.9
21 81.35
22 82.12
23 82.67
24 83.09
25 82.37
26 82.06
27 82.31
28 82.06
29 83.09
30 84.37
};
\addlegendentry{$\mathtt{HierSGD}$}
\addplot [semithick, darkorange25512714, mark=diamond*, mark size=1, mark options={solid}]
table {%
1 43.39
2 65.25
3 68.29
4 66.12
5 69.45
6 71.28
7 73.11
8 72.35
9 76.57
10 77.71
11 77.46
12 77.76
13 79.07
14 76.31
15 79.33
16 80.5
17 78.54
18 81.46
19 81.56
20 82.43
21 82.32
22 81.23
23 82.09
24 82.74
25 83.43
26 82.92
27 83.84
28 83.95
29 83.52
30 84.26
};
\addlegendentry{$\mathtt{Hier\text{-}Local\text{-}QSGD}$}
\addplot [semithick, forestgreen4416044, mark=*, mark size=1, mark options={solid}]
table {%
1 66.95
2 69.4
3 71.81
4 72.63
5 71.22
6 74.08
7 73.23
8 75.44
9 73.96
10 76.59
11 75.43
12 78.07
13 76.7
14 77.28
15 77.69
16 78.58
17 78.85
18 80.07
19 78.65
20 79.87
21 80.86
22 81.41
23 81.23
24 81.76
25 81.48
26 82.34
27 82.22
28 82.48
29 82.95
30 83.21
};
\addlegendentry{$\mathtt{HierSignSGD}$}
\addplot [semithick, crimson2143940, mark=triangle*, mark size=1, mark options={solid}]
table {%
1 67.13
2 66.71
3 69.18
4 71.7
5 71.81
6 73.01
7 73.28
8 75.2
9 74.6
10 76.61
11 75.12
12 77.79
13 76.7
14 78.87
15 77.11
16 79.29
17 79.49
18 79.83
19 81.24
20 80.23
21 80.73
22 81.45
23 82.21
24 82.08
25 82.38
26 82.18
27 82.6
28 82.45
29 82.61
30 83.32
};
\addlegendentry{$\mathtt{DC\text{-}HierSignSGD}$}
\end{axis}

\end{tikzpicture}}%
        \end{minipage}
        \label{fig:SignvsSGD_FMNIST_iid}
    }

    \vspace{0.2cm}

    \subfloat[Fashion-MNIST, non-IID ($\alpha = 0.1$)]{%
        \begin{minipage}{\SCALE\textwidth}
            \centering
            \resizebox{\linewidth}{!}{% This file was created with tikzplotlib v0.10.1.post13.
\begin{tikzpicture}[scale=0.47]

\definecolor{crimson2143940}{RGB}{214,39,40}
\definecolor{darkgrey176}{RGB}{176,176,176}
\definecolor{darkorange25512714}{RGB}{255,127,14}
\definecolor{forestgreen4416044}{RGB}{44,160,44}
\definecolor{lightgrey204}{RGB}{204,204,204}
\definecolor{steelblue31119180}{RGB}{31,119,180}

\begin{axis}[
legend cell align={left},
legend style={
  fill opacity=0.8,
  draw opacity=1,
  text opacity=1,
  at={(1,0.0)},
  anchor=south east,
  draw=lightgrey204
},
tick align=outside,
tick pos=left,
x grid style={darkgrey176},
xlabel={Global Round},
xmajorgrids,
xmin=-0.45, xmax=31.45,
xtick style={color=black},
y grid style={darkgrey176},
ylabel={Test Accuracy (\%)},
ymajorgrids,
ymin=12.6745, ymax=83.6355,
ytick style={color=black},
ytick distance = 8
]
\addplot [semithick, steelblue31119180, mark=square*, mark size=1, mark options={solid}]
table {%
1 19
2 43.9
3 50.52
4 56.44
5 59.66
6 62.57
7 61.83
8 68.14
9 67.08
10 71.39
11 69.66
12 73.25
13 69.49
14 72.83
15 72.48
16 74.32
17 73.47
18 75.17
19 75.52
20 75.25
21 76.2
22 75.59
23 77.1
24 76.85
25 77.7
26 76.42
27 77.97
28 78.03
29 78.68
30 79.27
};
\addlegendentry{$\mathtt{HierSGD}$}
\addplot [semithick, darkorange25512714, mark=diamond*, mark size=1, mark options={solid}]
table {%
1 22.95
2 35.82
3 51
4 55.14
5 59.28
6 59.71
7 64.35
8 69.69
9 69.3
10 70.77
11 70.49
12 73.2
13 72.75
14 69.14
15 74.26
16 71.57
17 74.37
18 75.58
19 74.76
20 73.73
21 76.05
22 76.68
23 77.97
24 77.61
25 76.96
26 77.14
27 78.36
28 76.84
29 76.85
30 78.19
};
\addlegendentry{$\mathtt{Hier\text{-}Local\text{-}QSGD}$}
\addplot [semithick, forestgreen4416044, mark=*, mark size=1, mark options={solid}]
table {%
1 25.17
2 36.18
3 43.56
4 49.75
5 52.56
6 55.73
7 54.48
8 56.11
9 53.79
10 57.79
11 57.48
12 57.02
13 59.71
14 57.13
15 60.45
16 61.1
17 58.29
18 63.69
19 64.56
20 66.1
21 64.49
22 66.02
23 66.48
24 66.43
25 64.46
26 62.09
27 64.55
28 65.65
29 64.99
30 65.24
};
\addlegendentry{$\mathtt{HierSignSGD}$}
\addplot [semithick, crimson2143940, mark=triangle*, mark size=1, mark options={solid}]
table {%
1 31.79
2 34.61
3 47.46
4 47.88
5 60.55
6 64.78
7 67.29
8 68.13
9 68.99
10 69.74
11 70.67
12 74.49
13 75.07
14 75.1
15 74.82
16 76.24
17 76.77
18 76.73
19 75.88
20 78.27
21 78.32
22 78.46
23 78.95
24 79.42
25 76.85
26 79.8
27 80.35
28 76.68
29 80.41
30 80.23
};
\addlegendentry{$\mathtt{DC\text{-}HierSignSGD}$}
\end{axis}

\end{tikzpicture}}%
        \end{minipage}
        \label{fig:SignvsSGD_FMNIST_non_iid}
    }
    \figspace
    \subfloat[CIFAR-10, IID]{%
        \begin{minipage}{\SCALE\textwidth}
            \centering
            \resizebox{\linewidth}{!}{% This file was created with tikzplotlib v0.10.1.post13.
\begin{tikzpicture}[scale=0.47]

\definecolor{crimson2143940}{RGB}{214,39,40}
\definecolor{darkgrey176}{RGB}{176,176,176}
\definecolor{darkorange25512714}{RGB}{255,127,14}
\definecolor{forestgreen4416044}{RGB}{44,160,44}
\definecolor{lightgrey204}{RGB}{204,204,204}
\definecolor{steelblue31119180}{RGB}{31,119,180}

\begin{axis}[
legend cell align={left},
legend style={
  fill opacity=0.8,
  draw opacity=1,
  text opacity=1,
  at={(1,0.0)},
  anchor=south east,
  draw=lightgrey204
},
tick align=outside,
tick pos=left,
x grid style={darkgrey176},
xlabel={Global Round},
xmajorgrids,
xmin=0, xmax=51.95,
xtick style={color=black},
y grid style={darkgrey176},
ylabel={Test Accuracy (\%)},
ymajorgrids,
ymin=9.03, ymax=70.19,
ytick style={color=black},
ytick distance = 6
]
\addplot [semithick, steelblue31119180, mark=square*, mark size=1, mark options={solid}]
table {%
1 24.35
2 31.43
3 41.55
4 44.9
5 45.39
6 47.95
7 45.77
8 44.66
9 49.47
10 50.62
11 51.58
12 52.84
13 53.7
14 53.62
15 53.64
16 54.08
17 54.87
18 55.49
19 55.29
20 57.21
21 56.28
22 55.04
23 58.95
24 58.68
25 58.44
26 60.05
27 58.97
28 58.42
29 59.29
30 60.57
31 60.35
32 61.42
33 61.05
34 61.92
35 60
36 62.08
37 62.3
38 62.51
39 62.77
40 61.18
41 62.5
42 62.86
43 62.03
44 63.64
45 64.01
46 63.52
47 64.07
48 63.13
49 62.61
50 64.28
};
\addlegendentry{$\mathtt{HierSGD}$}
\addplot [semithick, darkorange25512714, mark=diamond*, mark size=1, mark options={solid}]
table {%
1 11.81
2 12.7
3 24.36
4 36.17
5 37.64
6 39.62
7 41.55
8 42.17
9 43
10 42.86
11 44.88
12 44.6
13 44.35
14 46.41
15 47.02
16 48.13
17 48.78
18 48.79
19 49.8
20 49.46
21 48.91
22 50.33
23 50.99
24 50.94
25 49.43
26 52.52
27 52.96
28 53.16
29 52.9
30 53.03
31 54.17
32 53.08
33 54.96
34 54.07
35 54.71
36 54.22
37 54.93
38 55.79
39 56.39
40 56.1
41 54.74
42 56.25
43 56.48
44 57.1
45 56.61
46 56.88
47 57.59
48 57.93
49 57.58
50 57.89
};
\addlegendentry{$\mathtt{Hier\text{-}Local\text{-}QSGD}$}
\addplot [semithick, forestgreen4416044, mark=*, mark size=1, mark options={solid}]
table {%
1 30.44
2 40.27
3 40.22
4 45.83
5 48.68
6 50.07
7 53.84
8 56
9 56.21
10 55.17
11 57.47
12 60.13
13 58.43
14 61.77
15 61.82
16 61.56
17 62.85
18 63.07
19 63.54
20 64.42
21 63.63
22 63.77
23 62.74
24 64.85
25 65.52
26 65.51
27 64.21
28 65.11
29 66.04
30 65.29
31 63.89
32 65.21
33 65.79
34 65.91
35 65.65
36 66.21
37 66.28
38 66.58
39 65.5
40 64.69
41 66.43
42 66.24
43 65.14
44 66.07
45 66.68
46 66.4
47 65.95
48 65.17
49 66.31
50 65.87
};
\addlegendentry{$\mathtt{HierSignSGD}$}
\addplot [semithick, crimson2143940, mark=triangle*, mark size=1, mark options={solid}]
table {%
1 33.28
2 40.66
3 42.94
4 36.14
5 44.9
6 53.54
7 55.65
8 56.23
9 56.75
10 59.27
11 58.24
12 59.37
13 60.56
14 61.79
15 61.85
16 62.38
17 61.52
18 61.92
19 63.18
20 62.91
21 62.68
22 62.96
23 64.17
24 63.79
25 64.9
26 64.85
27 65.38
28 65.62
29 65.6
30 65.43
31 66.06
32 65.63
33 65.56
34 65.69
35 66
36 65.7
37 65.91
38 63.85
39 66.32
40 66.73
41 66.71
42 66.38
43 66.55
44 66.34
45 66.62
46 66.65
47 66.09
48 66.12
49 66.19
50 65.72
};
\addlegendentry{$\mathtt{DC\text{-}HierSignSGD}$}
\end{axis}

\end{tikzpicture}}%
        \end{minipage}
        \label{fig:SignvsSGD_CIFAR10_iid}
    }
    \figspace
    \subfloat[CIFAR-10, non-IID ($\alpha = 0.1$)]{%
        \begin{minipage}{\SCALE\textwidth}
            \centering
            \resizebox{\linewidth}{!}{% This file was created with tikzplotlib v0.10.1.post13.
\begin{tikzpicture}[scale=0.47]

\definecolor{crimson2143940}{RGB}{214,39,40}
\definecolor{darkgrey176}{RGB}{176,176,176}
\definecolor{darkorange25512714}{RGB}{255,127,14}
\definecolor{forestgreen4416044}{RGB}{44,160,44}
\definecolor{lightgrey204}{RGB}{204,204,204}
\definecolor{steelblue31119180}{RGB}{31,119,180}

\begin{axis}[
legend cell align={left},
legend style={
  fill opacity=0.8,
  draw opacity=1,
  text opacity=1,
  at={(1,0.0)},
  anchor=south east,
  draw=lightgrey204
},
tick align=outside,
tick pos=left,
x grid style={darkgrey176},
xlabel={Global Round},
xmajorgrids,
xmin=0, xmax=51.95,
xtick style={color=black},
y grid style={darkgrey176},
ylabel={Test Accuracy (\%)},
ymajorgrids,
ymin=7.7375, ymax=57.5125,
ytick style={color=black},
ytick distance = 6
]
\addplot [semithick, steelblue31119180, mark=square*, mark size=1, mark options={solid}]
table {%
1 12.01
2 21.69
3 24.69
4 27.33
5 27.74
6 27.12
7 30.85
8 34.51
9 32.86
10 34.96
11 36.33
12 37.05
13 38.41
14 36.7
15 39.52
16 39.01
17 40.13
18 41.17
19 41.45
20 41.34
21 42.11
22 43.39
23 43.79
24 43.36
25 40.66
26 44.96
27 45.53
28 46.09
29 43.98
30 44.31
31 45.35
32 43.96
33 45.39
34 45.56
35 47.64
36 48.81
37 46.74
38 47.22
39 47.5
40 47.1
41 48.4
42 49.03
43 47.77
44 48.11
45 45.94
46 47.73
47 46.49
48 48.2
49 49.43
50 47.71
};
\addlegendentry{$\mathtt{HierSGD}$}
\addplot [semithick, darkorange25512714, mark=diamond*, mark size=1, mark options={solid}]
table {%
1 10
2 11.07
3 17.26
4 19.86
5 26.01
6 26.37
7 30.21
8 28.45
9 32.97
10 31.45
11 30.35
12 33.65
13 34.31
14 36.97
15 36.77
16 34.89
17 37.21
18 38
19 36.75
20 36.28
21 37.88
22 33.93
23 38.37
24 37.57
25 36.82
26 37.1
27 39.08
28 41.48
29 40.2
30 38.09
31 38.8
32 40.41
33 39.68
34 39.69
35 40.93
36 41.06
37 38.41
38 39.1
39 40.58
40 41.29
41 41.73
42 41.64
43 41.56
44 40.84
45 40.13
46 41.99
47 41.18
48 40.54
49 41.98
50 39.27
};
\addlegendentry{$\mathtt{Hier\text{-}Local\text{-}QSGD}$}
\addplot [semithick, forestgreen4416044, mark=*, mark size=1, mark options={solid}]
table {%
1 18.23
2 23.37
3 20.44
4 22
5 28.79
6 34.64
7 32.23
8 31.59
9 33.48
10 38.03
11 35.86
12 37.02
13 37.59
14 37.78
15 39.97
16 41.76
17 38.84
18 41.93
19 41.71
20 43.16
21 43.57
22 44.12
23 43.39
24 43.91
25 45.63
26 45.28
27 46.93
28 45.7
29 44.27
30 43.88
31 45.01
32 47.18
33 45.55
34 44.4
35 46.41
36 45.22
37 45.26
38 45.46
39 48.15
40 47.28
41 45.5
42 47.03
43 45.9
44 47.35
45 45.68
46 46.28
47 49.63
48 46.13
49 47
50 48.85
};
\addlegendentry{$\mathtt{HierSignSGD}$}
\addplot [semithick, crimson2143940, mark=triangle*, mark size=1, mark options={solid}]
table {%
1 15.41
2 22.72
3 28.77
4 30.4
5 33.65
6 37.14
7 36.31
8 34.84
9 38.64
10 42.38
11 44.57
12 41.77
13 47.17
14 46.46
15 46.87
16 45.2
17 44.68
18 48.03
19 48.6
20 45.76
21 44.97
22 48.26
23 49.13
24 51.23
25 51.5
26 49.43
27 49.15
28 50.82
29 51.52
30 50
31 50.83
32 51.41
33 52.62
34 51.77
35 50.63
36 51.44
37 53.47
38 52.45
39 52.52
40 51.35
41 52.01
42 52.83
43 52.1
44 51.25
45 54.09
46 52.73
47 52.32
48 53.56
49 52
50 54.42
};
\addlegendentry{$\mathtt{DC\text{-}HierSignSGD}$}
\end{axis}

\end{tikzpicture}}%
        \end{minipage}
        \label{fig:SignvsSGD_CIFAR10_non_iid}
    }
   \caption{Test accuracy comparison of the proposed sign-based methods with full-precision
and quantized baselines on EMNIST-Digits, Fashion-MNIST, and CIFAR-10, using
batch size \(B=400\). For EMNIST-Digits, the step-sizes are \(\mu=1\) for the
SGD-based baselines and \(\mu=5\times 10^{-3}\) for the sign-based methods,
with correction strength \(\rho=0.2\). For Fashion-MNIST, these values are
\(\mu=0.06\), \(\mu=3\times 10^{-4}\), and \(\rho=0.07\), respectively.
For CIFAR-10, we use the decaying step-size \(\mu_t=\mu_0/\sqrt{t+1}\), with
\(\mu_0=0.08\) for the SGD-based baselines and \(\mu_0=10^{-3}\) for the
sign-based methods.}
    \label{fig:Comparison}
\end{figure*}

\begin{figure}[!t]
    \centering
    \subfloat[IID distribution]{%
        % This file was created with tikzplotlib v0.10.1.
\begin{tikzpicture}[scale=\SCALETWO]

\definecolor{crimson2143940}{RGB}{214,39,40}
\definecolor{darkgrey176}{RGB}{176,176,176}
\definecolor{darkorange25512714}{RGB}{255,127,14}
\definecolor{forestgreen4416044}{RGB}{44,160,44}
\definecolor{lightgrey204}{RGB}{204,204,204}
\definecolor{steelblue31119180}{RGB}{31,119,180}

\begin{axis}[
legend cell align={left},
legend columns=2,
legend style={
  fill opacity=0.8,
  draw opacity=1,
  text opacity=1,
  at={(1,1)},
  anchor=north east,
  draw=lightgrey204,
  /tikz/every even column/.append style={column sep=0.1cm}
},
legend style={fill opacity=0.8, draw opacity=1, text opacity=1, draw=lightgrey204},
log basis y={10},
tick align=outside,
tick pos=left,
x grid style={darkgrey176},
xlabel={Global Round},
xmajorgrids,
xmin=-1, xmax=21,
xminorgrids,
xtick style={color=black},
y grid style={darkgrey176},
ylabel={Training Loss},
ymajorgrids,
yminorgrids,
ymode=log,
log basis y={10},
ymin=0.11, ymax=3,
ytick style={color=black},
ytick={0.01,0.1,1,10,100},
yticklabels={
  \(\displaystyle 10^{-2}\),
  \(\displaystyle 10^{-1}\),
  \(\displaystyle 10^{0}\),
  \(\displaystyle 10^{1}\),
  \(\displaystyle 10^{2}\)
},
minor ytick={0.2,0.3,0.4,0.5,0.6,0.7,0.8,0.9,2,3},
yminorgrids,
]
\addplot [thick, steelblue31119180, mark=*, mark size=2, mark repeat=5, mark options={solid}]
table {%
0 2.50071094551086
1 1.44848852113088
2 1.04795084514618
3 0.824764854685465
4 0.662546654224396
5 0.557175817553202
6 0.481099123287201
7 0.416975639931361
8 0.368631569163005
9 0.335709554672241
10 0.308994230063756
11 0.291390423727035
12 0.274523572158813
13 0.258853594319026
14 0.247515315882365
15 0.237091263214747
16 0.227709850676854
17 0.221937347777685
18 0.216025619745255
19 0.210543360082308
20 0.203200821757317
};
\addlegendentry{\(\displaystyle T_E=10\)}
\addplot [thick, steelblue31119180, dashed, forget plot]
table {%
0 2.50071094551086
1 1.43726860815684
2 1.02101215101878
3 0.797312611643473
4 0.642184805806478
5 0.540303168201447
6 0.462581720336278
7 0.404645726013184
8 0.363102971140544
9 0.330351880772909
10 0.306886617294947
11 0.29092308417956
12 0.274367337846756
13 0.260464875284831
14 0.243326450856527
15 0.232690269486109
16 0.224825195423762
17 0.218352219518026
18 0.211568037565549
19 0.207070597513517
20 0.202569420083364
};
\addplot [thick, darkorange25512714, mark=*, mark size=2, mark repeat=5, mark options={solid}]
table {%
0 2.50071094551086
1 1.20085433190664
2 0.821071600341797
3 0.602931565189362
4 0.463645206928253
5 0.387356567557653
6 0.336046543550491
7 0.302307763719559
8 0.273303000068665
9 0.253423508834839
10 0.237005658658346
11 0.224656839084625
12 0.215545798865954
13 0.206466865372658
14 0.200062532973289
15 0.192809411946932
16 0.187426705447833
17 0.181909607068698
18 0.177028699239095
19 0.173564196387927
20 0.169448789286613
};
\addlegendentry{\(\displaystyle T_E=15\)}
\addplot [thick, darkorange25512714, dashed, forget plot]
table {%
0 2.50071094551086
1 1.19238414745331
2 0.802286587238312
3 0.59453381169637
4 0.474573833465576
5 0.388780689366659
6 0.338237266047796
7 0.297510080512365
8 0.275912928342819
9 0.251723586352666
10 0.240214793157578
11 0.225973152605693
12 0.216681691296895
13 0.206892894617716
14 0.199147137069702
15 0.193869067788124
16 0.186926081601779
17 0.181430969293912
18 0.177055894843737
19 0.173330140248934
20 0.167662418850263
};
\addplot [thick, forestgreen4416044, mark=*, mark size=2, mark repeat=5, mark options={solid}]
table {%
0 2.50071094551086
1 1.03772761077881
2 0.663129014841715
3 0.477524752521515
4 0.372828663778305
5 0.311740027793248
6 0.274721905438105
7 0.249341964912415
8 0.229845669031143
9 0.215855549120903
10 0.204463167262077
11 0.19561055161953
12 0.188513791910807
13 0.181281120355924
14 0.176252313748995
15 0.169570243477821
16 0.164892429892222
17 0.16041489200592
18 0.156392738246918
19 0.152408630649249
20 0.148855328623454
};
\addlegendentry{\(\displaystyle T_E=20\)}
\addplot [thick, forestgreen4416044, dashed, forget plot]
table {%
0 2.50071094551086
1 1.02797704423269
2 0.648811595948537
3 0.463425098196665
4 0.363658617448807
5 0.306266430775325
6 0.269850385665894
7 0.248798642873764
8 0.227855270147324
9 0.215036776852608
10 0.203617523201307
11 0.195154068811735
12 0.186751577941577
13 0.18026946665446
14 0.175136099425952
15 0.169798940126101
16 0.16525129776001
17 0.160539046351115
18 0.156834811902046
19 0.153167048462232
20 0.15007592300574
};
\addplot [thick, crimson2143940, mark=*, mark size=2, mark repeat=5, mark options={solid}]
table {%
0 2.50071094551086
1 0.916122624174754
2 0.55764197251002
3 0.390050844049454
4 0.311275624418259
5 0.266749067322413
6 0.238290225601196
7 0.219113608980179
8 0.205725696873665
9 0.193647624937693
10 0.185624110221863
11 0.176083390037219
12 0.170401674199104
13 0.163797483007113
14 0.159052474157015
15 0.153761694041888
16 0.149155352338155
17 0.145788149762154
18 0.142733155091604
19 0.138711228617032
20 0.136460988156001
};
\addlegendentry{\(\displaystyle T_E=25\)}
\addplot [thick, crimson2143940, dashed, forget plot]
table {%
0 2.50071094551086
1 0.907956110604604
2 0.544763383483887
3 0.383844449345271
4 0.308107870149612
5 0.263914059527715
6 0.236853350178401
7 0.218181143379211
8 0.204942380897204
9 0.193180599967639
10 0.184083939361572
11 0.175993859998385
12 0.169092675081889
13 0.162680897116661
14 0.159342199921608
15 0.153393354860942
16 0.149459905838966
17 0.144762698372205
18 0.140400572768847
19 0.137656455294291
20 0.134975646813711
};

\end{axis}

\end{tikzpicture}%
        \label{fig:T_edge_IID}
    }
    \vspace{0.15cm}
    \subfloat[Non-IID distribution with $\alpha = 0.1$]{%
        % This file was created with tikzplotlib v0.10.1.
\begin{tikzpicture}[scale=\SCALETWO]

\definecolor{crimson2143940}{RGB}{214,39,40}
\definecolor{darkgrey176}{RGB}{176,176,176}
\definecolor{darkorange25512714}{RGB}{255,127,14}
\definecolor{forestgreen4416044}{RGB}{44,160,44}
\definecolor{lightgrey204}{RGB}{204,204,204}
\definecolor{steelblue31119180}{RGB}{31,119,180}

\begin{axis}[
legend cell align={left},
legend columns=2,
legend style={
  fill opacity=0.8,
  draw opacity=1,
  text opacity=1,
  at={(1,1)},
  anchor=north east,
  draw=lightgrey204,
  /tikz/every even column/.append style={column sep=0.1cm}
},
legend style={fill opacity=0.8, draw opacity=1, text opacity=1, draw=lightgrey204},
log basis y={10},
tick align=outside,
tick pos=left,
x grid style={darkgrey176},
xlabel={Global Round},
xmajorgrids,
xmin=-1, xmax=21,
xminorgrids,
xtick style={color=black},
y grid style={darkgrey176},
ylabel={Training Loss},
ymajorgrids,
yminorgrids,
ymode=log,
ymin=0.19259874520969, ymax=3,
ytick style={color=black},
% Major ticks: only powers of 10
ytick={0.01,0.1,1,10,100},
yticklabels={
  \(\displaystyle {10^{-2}}\),
  \(\displaystyle {10^{-1}}\),
  \(\displaystyle {10^{0}}\),
  \(\displaystyle {10^{1}}\),
  \(\displaystyle {10^{2}}\)
},
minor ytick={0.2,0.3,0.4,0.5,0.6,0.7,0.8,0.9,2},
]
\addplot [thick, steelblue31119180, mark=*, mark size=2, mark repeat=5, mark options={solid}]
table {%
0 2.50071094551086
1 1.97353654219309
2 1.47725534267426
3 1.21930975335439
4 1.03687287082672
5 0.887117255624135
6 0.77679499130249
7 0.687253788979848
8 0.609643005847931
9 0.549210343805949
10 0.505574044768016
11 0.464431065734227
12 0.426220889918009
13 0.400962379391988
14 0.379210863447189
15 0.360381496413549
16 0.340463346115748
17 0.329545010519028
18 0.321698059956233
19 0.30067194237709
20 0.293904409837723
};
\addlegendentry{\(\displaystyle T_E=10\)}
\addplot [thick, steelblue31119180, dashed, forget plot]
table {%
0 2.50071094551086
1 1.96076586208343
2 1.76497220719655
3 1.62888672281901
4 1.51890773054759
5 1.38040919488271
6 1.28863864955902
7 1.22018849493663
8 1.17261806329091
9 1.12489710795085
10 1.08600631561279
11 1.05220100536346
12 0.999713766733805
13 0.973938679250081
14 0.955569016774495
15 0.92731325229009
16 0.904923027801514
17 0.87586062186559
18 0.852749798488617
19 0.815884001191457
20 0.803934013621012
};
\addplot [thick, darkorange25512714, mark=*, mark size=2, mark repeat=5, mark options={solid}]
table {%
0 2.50071094551086
1 1.91434154911041
2 1.37773001397451
3 1.02198114287059
4 0.811338956801097
5 0.658548867289225
6 0.56481009982427
7 0.475954374663035
8 0.429761964813868
9 0.381465673589706
10 0.355808446772893
11 0.330225502411524
12 0.310494446897507
13 0.294082592519124
14 0.284292123842239
15 0.271847816801071
16 0.271371105702718
17 0.254041790533066
18 0.250976918570201
19 0.244269498109818
20 0.234844075123469
};
\addlegendentry{\(\displaystyle T_E=15\)}
\addplot [thick, darkorange25512714, dashed, forget plot]
table {%
0 2.50071094551086
1 1.90840340938568
2 1.70346474959056
3 1.47499167283376
4 1.30572285060883
5 1.20160420341492
6 1.11749403998057
7 1.05152265008291
8 0.971943464374542
9 0.917598934173584
10 0.86805177230835
11 0.8299748655955
12 0.787405504449209
13 0.746551453049978
14 0.727228523763021
15 0.69571717710495
16 0.669242267131805
17 0.649981382401784
18 0.625309562651316
19 0.606295501168569
20 0.569917261091868
};
\addplot [thick, forestgreen4416044, mark=*, mark size=2, mark repeat=5, mark options={solid}]
table {%
0 2.50071094551086
1 1.91131441268921
2 1.26771877943675
3 0.917129479249318
4 0.693105842240651
5 0.561134817790985
6 0.462571389357249
7 0.401033201710383
8 0.361819298410416
9 0.322998783683777
10 0.306965618817012
11 0.283792404635747
12 0.273690796677272
13 0.2585048336188
14 0.245699498430888
15 0.24882724386851
16 0.239048372920354
17 0.242059449275335
18 0.219047735611598
19 0.239004272508621
20 0.217607058858871
};
\addlegendentry{\(\displaystyle T_E=20\)}
\addplot [thick, forestgreen4416044, dashed, forget plot]
table {%
0 2.50071094551086
1 1.88575260384878
2 1.64829942213694
3 1.33438170223236
4 1.21870577818553
5 1.11714293651581
6 1.0118430683136
7 0.932441371409098
8 0.869862279891968
9 0.800256318187714
10 0.762470199362437
11 0.699392076079051
12 0.668454707368215
13 0.64666717411677
14 0.602075250720978
15 0.577077322610219
16 0.576678049755096
17 0.545329683176676
18 0.547523098214467
19 0.511471947892507
20 0.50643132870992
};
\addplot [thick, crimson2143940, mark=*, mark size=2, mark repeat=5, mark options={solid}]
table {%
0 2.50071094551086
1 1.88223551019033
2 1.18264012031555
3 0.846209392642975
4 0.61684370953242
5 0.483113413747152
6 0.422919181934992
7 0.339329007991155
8 0.353542425409953
9 0.295665425157547
10 0.288506144984563
11 0.278966292699178
12 0.264830457639694
13 0.34622461810112
14 0.240052825180689
15 0.350074584007263
16 0.278315735642115
17 0.758439075374603
18 0.263582939386368
19 0.488010476589203
20 0.302235609817505
};
\addlegendentry{\(\displaystyle T_E=25\)}
\addplot [thick, crimson2143940, dashed, forget plot]
table {%
0 2.50071094551086
1 1.87857382818858
2 1.57161496842702
3 1.29155808550517
4 1.16578041413625
5 1.04412229162852
6 0.952052547868093
7 0.862354630533854
8 0.821755353577932
9 0.762355153369904
10 0.714355618953705
11 0.660401661554972
12 0.650832810688019
13 0.637981199741364
14 0.588640915648142
15 0.593285756142934
16 0.556483438014984
17 0.559234493128459
18 0.546254331302643
19 0.512554354317983
20 0.518445032278697
};
\end{axis}

\end{tikzpicture}%
        \label{fig:T_edge_nonIID}
    }
    \caption{Effect of $T_E$ on global training loss; comparing $\mathtt{DC\text{-}HierSignSGD}$ with $\rho=0.2$ (solid line) with $\mathtt{HierSignSGD}$ (dashed line).}
    \label{fig:T_edge}
\end{figure}

For EMNIST-Digits, we use a fully connected neural network with one hidden layer. For Fashion-MNIST, we use a convolutional neural network, while for CIFAR-10 we use ResNet-20 trained with a decaying step-size.

Simulations in our two-tier HFL setup use $Q=4$ edge servers and $|\mathcal{V}^q|=5$ devices per edge, totaling 20 participating devices.
Both IID and non-IID scenarios are considered. In the non-IID setting, statistical heterogeneity across edge servers is induced through a symmetric Dirichlet distribution, while keeping the devices within each edge IID to isolate the effect of edge-level skew. For each class $m$, a probability vector is sampled according to 
\[
\mathbf{p}_m \sim \text{Dirichlet}(\alpha \mathbf{1}_Q),
\] 
where $\alpha$ controls the concentration parameter. 
Each entry $[\mathbf{p}_m]_q$ represents the fraction of class-$m$ samples assigned to edge server $q$. Smaller $\alpha$ values yield imbalanced (non-IID) edge-level label distributions, while larger $\alpha$ values produce more uniform (IID-like) partitions. We set \(\alpha=0.1\), corresponding to a highly skewed inter-cluster data distribution and representing an extreme non-IID setting in our simulations. 
% Fig.~\ref{fig:histo} shows the distribution of digit labels across edge servers, clearly illustrating that the data partitioning is fairly heterogeneous.

%\vspace{-0.3cm}

\subsection{Learning Accuracy}

As a first scenario, we compare the proposed sign-based methods with relevant full-precision and quantized HFL baselines. The main full-precision baseline is an SGD-based hierarchical method, denoted by $\mathtt{HierSGD}$, in which devices transmit their full stochastic gradients to the edge server. The edge server then computes a weighted average of the received gradients and performs a standard gradient descent update. This method serves as the full-precision counterpart of the proposed sign-based algorithms. We also compare against the quantized HFL method proposed in~\cite{liu2022hierarchical}, where quantization is applied to the model parameters at two hierarchical layers rather than to the stochastic gradients. For a fair comparison, we assume ideal cloud--edge communication and apply an unbiased stochastic ternary quantizer to the device--edge model differences:
\begin{equation*} [\mathcal{Q}(\boldsymbol{\Delta})]_i = \begin{cases} \|\boldsymbol{\Delta}\|_2 \operatorname{sign}(\Delta_i), & \text{with probability } \dfrac{|\Delta_i|}{\|\boldsymbol{\Delta}\|_2}, \\[6pt] 0, & \text{otherwise}, \end{cases} \end{equation*} with \(\mathcal{Q}(\mathbf{0})=\mathbf{0}\). Compared with $\mathtt{HierSignSGD}$, this ternary quantizer has a higher communication cost. In particular, each device must transmit not only the signs of the selected coordinates, but also the support pattern indicating which coordinates are nonzero, together with the scaling factor \(\|\boldsymbol{\Delta}\|_2\).
To make the communication savings explicit, Table~\ref{tab:comm_cost} summarizes the device--edge uplink cost per global round for the considered methods. The table counts only device transmissions, since the device--edge uplink is the primary communication bottleneck in the considered hierarchical architecture. Full-precision quantities are represented using 32 bits per coordinate.
The hyperparameters of all methods are tuned empirically, guided by commonly used choices in the literature. In this experiment, we set \(T_E=15\).

\begin{table}[t] \centering 
\caption{Device--edge uplink cost per global round.} \label{tab:comm_cost} 
\scriptsize
\setlength{\tabcolsep}{3pt} 
\begin{tabular}{c|c|c|c|c} \hline Method & $\mathtt{HierSGD}$ & $\mathtt{Hier\text{-}Local\text{-}QSGD}$ & $\mathtt{HierSignSGD}$ & $\mathtt{DC\text{-}HierSignSGD}$ \\ \hline Bits/device & \(32T_Ed\) & \(> T_E(d+32)\) & \(T_Ed\) & \(T_Ed+32d\) \\ \hline \end{tabular} \end{table}

The results, depicted in Fig.~\ref{fig:Comparison}, highlight the benefit of the proposed correction mechanism. In the non-IID case, \(\mathtt{HierSignSGD}\) suffers from inter-cluster heterogeneity, which slows down convergence and degrades test accuracy. In contrast, \(\mathtt{DC\text{-}HierSignSGD}\) substantially mitigates this effect and achieves a more stable and accurate performance. This confirms that the proposed correction is effective when edge-level gradient dissimilarity is significant. In the IID case, the gap between the corrected and uncorrected sign-based methods is smaller, since the edge-level objectives are already well aligned with the global objective. For the EMNIST and CIFAR-10 datasets, the sign-based methods can outperform the full-precision SGD-based baseline, showing that binary device--edge communication does not necessarily compromise learning performance. Overall, \(\mathtt{DC\text{-}HierSignSGD}\) achieves either the best accuracy or performance comparable to the strongest baseline across the considered settings, while preserving the communication efficiency of sign-based training.

% This file was created with tikzplotlib v0.10.1.post13.
\begin{figure}[t] 
  \centering
\begin{tikzpicture}[scale=\SCALETWO]

\definecolor{crimson2143940}{RGB}{214,39,40}
\definecolor{darkgrey176}{RGB}{176,176,176}
\definecolor{darkorange25512714}{RGB}{255,127,14}
\definecolor{forestgreen4416044}{RGB}{44,160,44}
\definecolor{lightgrey204}{RGB}{204,204,204}
\definecolor{mediumpurple148103189}{RGB}{148,103,189}
\definecolor{steelblue31119180}{RGB}{31,119,180}

\begin{axis}[
legend cell align={left},
legend style={fill opacity=0.8, draw opacity=1, text opacity=1, draw=lightgrey204,at={(1,1)}},
log basis y={10},
tick align=outside,
tick pos=left,
x grid style={darkgrey176},
xlabel={Global Round},
xmajorgrids,
xmin=-1, xmax=21,
xminorgrids,
xtick style={color=black},
y grid style={darkgrey176},
ylabel={Training Loss},
ymajorgrids,
yminorgrids,
ymode=log,
ymin=0.232531080991499, ymax=2.80018335969894,
ytick style={color=black},
% Major ticks: only powers of 10
ytick={0.01,0.1,1,10,100},
yticklabels={
  \(\displaystyle {10^{-2}}\),
  \(\displaystyle {10^{-1}}\),
  \(\displaystyle {10^{0}}\),
  \(\displaystyle {10^{1}}\),
  \(\displaystyle {10^{2}}\)
},
% Minor ticks: shown but not labeled
minor ytick={0.3,0.4,0.5,0.6,0.7,0.8,0.9,2},
]
\addplot [thick, steelblue31119180, mark=*, mark size=2, mark repeat=5, mark options={solid}]
table {%
0 2.50071094551086
1 1.91434154911041
2 1.70673949883779
3 1.50857835127513
4 1.29597636102041
5 1.20524605960846
6 1.11946763515472
7 1.04140452721914
8 0.984344189453125
9 0.937158675924937
10 0.867610357570648
11 0.842168060557048
12 0.773834433778127
13 0.753552561950684
14 0.722246768697103
15 0.680107767995199
16 0.668144770145416
17 0.647879808203379
18 0.623026099681854
19 0.598549170017242
20 0.587118850930532
};
\addlegendentry{\(\displaystyle \rho=0\)}
\addplot [thick, darkorange25512714, mark=*, mark size=2, mark repeat=5, mark options={solid}]
table {%
0 2.50071094551086
1 1.91434154911041
2 1.4158823855718
3 1.16549396533966
4 0.975828320566813
5 0.845793090852102
6 0.745376352183024
7 0.651603367964427
8 0.594536990737915
9 0.548463694922129
10 0.494781586710612
11 0.467919578345617
12 0.431985944811503
13 0.412890583848953
14 0.393168220504125
15 0.367498648134867
16 0.35440019253095
17 0.338674572944641
18 0.333034664169947
19 0.321360284662247
20 0.297635556427638
};
\addlegendentry{\(\displaystyle \rho=0.05\)}
\addplot [thick, forestgreen4416044, mark=*, mark size=2, mark repeat=5, mark options={solid}]
table {%
0 2.50071094551086
1 1.91434154911041
2 1.34449975840251
3 1.07526753749847
4 0.884791590722402
5 0.742226424439748
6 0.64565551144282
7 0.556907317320506
8 0.497946707248688
9 0.457761702076594
10 0.411529930448532
11 0.391393985493978
12 0.358530989694595
13 0.343939359680812
14 0.327096394729614
15 0.327505658086141
16 0.298429971583684
17 0.295306274223328
18 0.283938252941767
19 0.276591233332952
20 0.260377819665273
};
\addlegendentry{\(\displaystyle \rho=0.10\)}
\addplot [thick, crimson2143940, mark=*, mark size=2, mark repeat=5, mark options={solid}]
table {%
0 2.50071094551086
1 1.91434154911041
2 1.32361266021729
3 1.01757487589518
4 0.815211142190297
5 0.669184738381704
6 0.571705088329315
7 0.472345964018504
8 0.438950133593877
9 0.359888503026962
10 0.355700594536463
11 0.313148961257935
12 0.321032923491796
13 0.291461223999659
14 0.281061869223913
15 0.260956506999334
16 0.306584726842244
17 0.260581944910685
18 0.31470831211408
19 0.269536208804448
20 0.296988886435827
};
\addlegendentry{\(\displaystyle \rho=0.40\)}
\addplot [thick, mediumpurple148103189, mark=*, mark size=2, mark repeat=5, mark options={solid}]
table {%
0 2.50071094551086
1 1.91434154911041
2 1.43225493164063
3 1.19189284534454
4 0.983708661651611
5 0.868721098804474
6 0.745235546080271
7 0.687159583822886
8 0.550125505034129
9 0.517004580974579
10 0.480501608498891
11 0.457239681609472
12 0.442957230488459
13 0.406545532226562
14 0.382886571741104
15 0.380851326004664
16 0.341205800867081
17 0.332695929924647
18 0.334154336913427
19 0.304459133402507
20 0.296046527814865
};
\addlegendentry{\(\displaystyle \rho=0.70\)}
\end{axis}

\end{tikzpicture}

\caption{Training sensitivity to different values of $\rho$ when $T_E=15$.}
\label{fig:rho_nonIID}
\end{figure}

One possible explanation for the strong IID performance of the sign-based methods on EMNIST-Digits is the noise structure encountered when training neural networks on digit-recognition datasets. Prior works have shown that stochastic-gradient noise in neural networks trained on MNIST-type datasets can exhibit heavy-tailed behavior~\cite{gurbuzbalaban2021heavy,csimcsekli2019heavy}. In such regimes, gradient magnitudes may fluctuate strongly, while coordinate-wise directional information remains comparatively more reliable. Since sign-based methods discard magnitude information and use only the gradient direction, they can be less sensitive to magnitude-induced fluctuations than standard SGD. This provides an intuitive explanation for why \(\mathtt{HierSignSGD}\) can outperform the full-precision \(\mathtt{HierSGD}\) baseline in the IID EMNIST setting. However, under non-IID edge-level partitions, local gradients may become systematically biased toward different objectives, and the robustness of the sign operator to magnitude noise alone is no longer sufficient. This is precisely where the proposed drift correction becomes beneficial.
% As a result, empirical evaluation becomes essential to determine which method performs better under non-IID conditions.

\subsection{Effect of $T_E$}

The impact of the number of local steps $T_E$ depends strongly on the degree of data heterogeneity. In the IID case, the difference between \(\mathtt{DC\text{-}HierSignSGD}\) and \(\mathtt{HierSignSGD}\) is relatively small, as shown in Fig.~\ref{fig:T_edge_IID}. This is expected, since the edge-level objectives are nearly aligned with the global objective. We observe that the dominant effect of increasing \(T_E\) is that more local sign updates are performed between cloud aggregations, leading to faster loss reduction.

The behavior is markedly different in the non-IID case. Here, the dashed curves corresponding to the uncorrected method exhibit slower loss reduction, while the solid curves show that the proposed correction substantially improves convergence. This confirms that the correction is most effective when inter-cluster gradient dissimilarity is significant. However, the dependence on \(T_E\) is not monotonic. Larger values of \(T_E\) can accelerate early progress by allowing more corrected local updates per global round, but they also increase the mismatch between the current local model and the point at which the correction was computed. This explains the oscillatory behavior observed for larger \(T_E\), especially in later global rounds. As discussed in the previous section, the parameters \(\rho\) and \(T_E\) are coupled and should be tuned jointly. We therefore next study the sensitivity of \(\mathtt{DC\text{-}HierSignSGD}\) to different values of \(\rho\).

\subsection{Sensitivity to $\rho$}

Fig.~\ref{fig:rho_nonIID} illustrates the sensitivity of \(\mathtt{DC\text{-}HierSignSGD}\) in the non-IID case to the correction strength \(\rho\) when \(T_E=15\). The case \(\rho=0\) corresponds to the uncorrected sign-based update and exhibits slower loss reduction than all nonzero values of \(\rho\), confirming the benefit of the proposed correction under inter-cluster heterogeneity. However, the effect of \(\rho\) is not monotonic. In this setting, moderate values, such as \(\rho =0.1\), provide the most stable decrease in training loss, whereas overly large correction strengths can introduce oscillations, especially in later global rounds, and may even degrade overall performance. This behavior is expected because, as training progresses and gradients become smaller, an aggressive correction term may dominate the stochastic gradient before the sign operation, leading to abrupt coordinate flips. Therefore, \(\rho\) controls a stability--correction tradeoff and must be tuned carefully in practice.

\section{Conclusion} \label{sec:conclusion}

In this paper, we developed a sign-based HFL framework that achieves stringent device--edge communication efficiency while addressing inter-cluster heterogeneity. Our analysis showed that, unlike HFL schemes based on full-precision SGD updates or conventional quantization, the direction-only nature of sign updates creates a persistent heterogeneity-induced drift term in the convergence bound. Since the sign operator discards magnitude information and interacts nonlinearly with edge-level objectives, this drift cannot be eliminated through parameter tuning alone.
To mitigate this issue, we proposed a drift-corrected sign-based algorithm in which devices apply a cloud-assisted gradient correction before transmitting binary signs to the edge server. The proposed correction reduces heterogeneity-induced drift while preserving binary device--edge communication during local training. We also extended the convergence result to the majority-vote setting. Numerical experiments confirmed that the corrected method improves the stability and accuracy of sign-based HFL under strong inter-cluster heterogeneity and can achieve performance comparable to full-precision hierarchical SGD with substantially lower device--edge communication.

\appendices

\section{Proof of Theorem~\ref{thm:1st}} \label{app:thm1}

Since there is only one device per server, we can drop the index $k$ for local variables. We start by expressing the $(t+1)$th global average as
\begin{align} \label{eq: w_update}
    \mathbf{w}^{(t+1)} &= \sum_{q=1}^Q \frac{D_q}{N}\mathbf{v}_{q}^{(t,T_E)}\nonumber \\ &= \sum_{q=1}^Q \frac{D_q}{N} \Big(\mathbf{v}_{q}^{(t,0)} -  \mu\sum_{\tau=0}^{T_E-1}  \operatorname{sgn}\!\big(\hat{\mathbf{g}}_{q}^{(t,\tau)}\big)\Big) \nonumber\\
    &= \sum_{q=1}^Q \frac{D_q}{N} \Big(\mathbf{w}^{(t)} -  \mu\sum_{\tau=0}^{T_E-1}  \operatorname{sgn}\!\big(\hat{\mathbf{g}}_{q}^{(t,\tau)}\big)\Big) \nonumber\\
    &= \mathbf{w}^{(t)} - \mu\sum_{q=1}^Q \sum_{\tau=0}^{T_E-1}  \frac{D_q}{N}\operatorname{sgn}\!\big(\hat{\mathbf{g}}_{q}^{(t,\tau)}\big).
\end{align}
Employing assumption A2 with $\ell_\infty$ norm, we write
\begin{align}
    \mathcal{F}(\mathbf{w}^{(t+1)}) &- \mathcal{F}(\mathbf{w}^{(t)}) 
    \leq \Big\langle \nabla \mathcal{F}(\mathbf{w}^{(t)}), \, \mathbf{w}^{(t+1)} - \mathbf{w}^{(t)} \Big\rangle 
    \nonumber \\
    &+\, \frac{L}{2} \left\| \mathbf{w}^{(t+1)} - \mathbf{w}^{(t)} \right\|_\infty^2 
    \label{eq:Lipschitz} \\[0.5em]
    = &- \Big\langle \nabla \mathcal{F}(\mathbf{w}^{(t)}), \, 
    \mu\sum_{q=1}^Q \sum_{\tau=0}^{T_E-1}  \frac{D_q}{N}\operatorname{sgn}\!\big(\hat{\mathbf{g}}_{q}^{(t,\tau)}\big) \Big\rangle
    \nonumber \\%[0.2em]
    &+ \frac{L}{2} 
    \Big\| \mu\sum_{q=1}^Q \sum_{\tau=0}^{T_E-1}  
    \frac{D_q}{N}\operatorname{sgn}\!\big(\hat{\mathbf{g}}_{q}^{(t,\tau)}\big)
    \Big\|_\infty^2.
    \label{eq:twoTerms}
\end{align}

We first bound the second term in~\eqref{eq:twoTerms} by using the triangle inequality as
\begin{align} \label{eq:bound_chi}
    &\frac{L}{2}\mu^2 \Big\| \sum_{q=1}^Q \sum_{\tau=0}^{T_E-1}  \frac{D_q}{N}\operatorname{sgn}\!\big(\hat{\mathbf{g}}_{q}^{(t,\tau)}\big) \Big\|_\infty^2 \nonumber\\
    \le\,\, &\frac{L}{2}\mu^2 \Big( \sum_{q=1}^{Q} \sum_{\tau=0}^{T_E-1} \frac{D_q}{N} \Big)^{2}
=\,\, \frac{L}{2}\big(\mu T_E\big)^2.
\end{align}

We now proceed to bound the more challenging first term of~\eqref{eq:twoTerms}. Let us rewrite the inner product as
\begin{align} \label{eq: 1stterm_rewrite}
    &- \Big\langle \nabla \mathcal{F}(\mathbf{w}^{(t)}), \, 
    \mu\sum_{q=1}^Q \sum_{\tau=0}^{T_E-1}  \frac{D_q}{N}\operatorname{sgn}\!\big(\hat{\mathbf{g}}_{q}^{(t,\tau)}\big) \Big\rangle \nonumber\\
    = &- \mu\sum_{q=1}^Q   \frac{D_q}{N}\sum_{\tau=0}^{T_E-1}\Big\langle \nabla \mathcal{F}(\mathbf{w}^{(t)}),\,\operatorname{sgn}\!\big(\hat{\mathbf{g}}_{q}^{(t,\tau)}\big)\Big\rangle.
\end{align}
For convenience, we temporarily adopt the following notational switches:
\begin{equation} \label{eq: notation_switch}
\begin{aligned}
   \nabla \mathcal{F}(\mathbf{w}^{(t)}) \,&\rightarrow\, \nabla\mathcal{F},
   \quad
   \nabla \mathcal{F}_q(\mathbf{w}^{(t)}) \,\rightarrow\, \nabla\mathcal{F}_q, \\
   \hat{\mathbf{g}}_{q}^{(t,\tau)} \,&\rightarrow\, \hat{\mathbf{g}}_q^{(\tau)},
   \quad {\mathbf{v}}_{q}^{(t,\tau)} \,\rightarrow\, {\mathbf{v}}_q^{(\tau)}.
\end{aligned}
\end{equation}
The goal is to bound $-\sum_{\tau=0}^{T_E-1}\big\langle \nabla \mathcal{F},\,\operatorname{sgn}\!\big(\hat{\mathbf{g}}_q^{(\tau)}\big)\big\rangle$. Let us denote this by $\Omega$ and  recast it as
\begin{align} \label{eq:Omega}
    \Omega \triangleq -\!\!\sum_{\tau=0}^{T_E-1}\Big\langle \nabla \mathcal{F},\,\operatorname{sgn}\!\big(\hat{\mathbf{g}}_q^{(\tau)}\big)\Big\rangle =\, \sum_{\tau=0}^{T_E-1} \Big(\!-\|\nabla \mathcal{F}\|_1 \nonumber \\ +2\sum_{i=1}^d\big|[\nabla \mathcal{F}]_i\big|\,\mathbb{I}\big\{
\underbrace{
\operatorname{sgn}([\nabla \mathcal{F}]_i)\neq 
\operatorname{sgn}([\hat{\mathbf{g}}_q^{(\tau)}]_i)
}_{\mathcal{A}_i^{(\tau)}}
\big\}\Big), 
\end{align}
where $\mathbb{I}\{\cdot\}$ denotes the indicator function. Observe that the anticipated progress in the convergence hinges on the falsity of the event inside the indicator function. For brevity, we have denoted this event by $\mathcal{A}_i^{(\tau)}$.
Taking the expectation conditioned on the previous iterate $\mathbf{w}^{(t)}$ yields
\begin{align} \label{eq:first_expected}
\mathbb{E}\{\Omega \mid \mathbf{w}^{(t)}\}
= \! \sum_{\tau=0}^{T_E-1}\! \Big(\!\! -\|\nabla \mathcal{F}\|_1 
  + 2\sum_{i=1}^d \bigl|[\nabla \mathcal{F}]_i\bigr|\,
\text{Pr}\big\{\mathcal{A}_i^{(\tau)}\big\} \Big).
\end{align}
Hence, if we expect the algorithm to converge, then the local gradients should at least be able to correctly estimate the sign of the global gradient with high probability. The main thing now is to bound this probability. To this end, we use the following relaxation
\begin{align} \label{eq:Markov's}
\text{Pr}\big\{\mathcal{A}_i^{(\tau)}\big\}
 &\leq \,\,\text{Pr}\left\{\big|[\hat{\mathbf{g}}_q^{(\tau)}]_i - [\nabla \mathcal{F}]_i\big| \geq \big|[\nabla \mathcal{F}]_i\big|\right\} \nonumber\\
&\leq\,\, \frac{\mathbb{E}\left\{\big|[\hat{\mathbf{g}}_q^{(\tau)}]_i - [\nabla \mathcal{F}]_i\big|\right\}}{\big|[\nabla \mathcal{F}]_i\big|},
\end{align}
where in the second line, we have employed Markov's inequality~\cite{BoucheronLugosiMassart2013}. Plugging the obtained result into~\eqref{eq:first_expected}, we get
\begin{align} \label{eq:main_eq}
    \eqref{eq:first_expected}\, &\le \sum_{\tau=0}^{T_E-1} \Big(\!-\|\nabla \mathcal{F}\|_1 + 2\sum_{i=1}^d\mathbb{E}\big\{\big|[\hat{\mathbf{g}}_q^{(\tau)}]_i - [\nabla \mathcal{F}]_i\big|\big\}\!\Big) \nonumber\\
    &= \sum_{\tau=0}^{T_E-1} \Big(\!-\|\nabla \mathcal{F}\|_1 \nonumber \\&+ 2\sum_{i=1}^d\mathbb{E}\big\{\big|[\hat{\mathbf{g}}_q^{(\tau)}]_i - [{\mathbf{g}}_q^{(\tau)}]_i + [{\mathbf{g}}_q^{(\tau)}]_i - [\nabla \mathcal{F}]_i\big|\big\}\!\Big) \nonumber \\
    &\le \sum_{\tau=0}^{T_E-1} \Big(\!-\|\nabla \mathcal{F}\|_1 + \frac{2\sigma d}{\sqrt{B}} \nonumber \\& \hspace{6em}+ 2\sum_{i=1}^d\mathbb{E}\big\{\big|[{\mathbf{g}}_q^{(\tau)}]_i - [\nabla \mathcal{F}]_i\big|\big\}\!\Big),
\end{align}
where we have added and subtracted the true local gradient components $[{\mathbf{g}}_q^{(\tau)}]_i$ in the third line, and in the last inequality, we have used the triangle inequality along with the mini-batch gradient assumption A3. 

Let us now define
\begin{align}
    \mathcal{I}^{(\tau)} \triangleq \sum_{i=1}^d\mathbb{E}\big\{\big|[{\mathbf{g}}_q^{(\tau)}]_i - [\nabla \mathcal{F}]_i\big|\big\}.
\end{align}
We attempt to bound $\mathcal{I}^{(\tau)}$ recursively. First, note that
\begin{subequations} \label{eq: I(0)}
\begin{align}
    \mathcal{I}^{(0)} &= \sum_{i=1}^d\mathbb{E}\Big\{\big|[{\mathbf{g}}_q^{(0)}]_i - [\nabla \mathcal{F}]_i\big|\Big\} \ \label{eq: I(0)_1}\\
     &= \sum_{i=1}^d\big|[\nabla \mathcal{F}_q]_i - [\nabla \mathcal{F}]_i\big| = \big\| \nabla \mathcal{F}_q - \nabla \mathcal{F} \big\|_1.  \label{eq: I(0)_2}
\end{align}
\end{subequations}
Thus, the base case is established. We next derive a recursive relation for $\mathcal{I}^{(\tau)}$
\begin{align*}
    \mathcal{I}^{(\tau)} &= \sum_{i=1}^d\mathbb{E}\Big\{\big|[{\mathbf{g}}_q^{(\tau)}]_i - [\nabla \mathcal{F}]_i\big|\Big\} \\
    &= \sum_{i=1}^d\mathbb{E}\Big\{\big|[{\mathbf{g}}_q^{(\tau-1)}]_i - [\nabla \mathcal{F}]_i + [{\mathbf{g}}_q^{(\tau)}]_i - [{\mathbf{g}}_q^{(\tau-1)}]_i\big|\Big\} \\
     &\le\, \sum_{i=1}^d\Bigl(\mathbb{E}\Big\{\big|[{\mathbf{g}}_q^{(\tau-1)}]_i - [\nabla \mathcal{F}]_i\big|\Big\} \\ &\hspace{8em}+ \mathbb{E}\Big\{\big|[{\mathbf{g}}_q^{(\tau)}]_i - [{\mathbf{g}}_q^{(\tau-1)}]_i\big|\Big\}\!\Bigr)  \\
     &=\, \mathcal{I}^{(\tau-1)} 
    + \mathbb{E}\Big\{\big\|\nabla \mathcal{F}_q(\mathbf{v}_q^{(\tau)}) - \nabla \mathcal{F}_q(\mathbf{v}_q^{(\tau-1)}) \big\|_1\Big\} \\
    &\le \mathcal{I}^{(\tau-1)} 
    + \mathbb{E}\Big\{L\big \|\mathbf{v}_q^{(\tau)} - \mathbf{v}_q^{(\tau-1)} \big\|_\infty\Big\} \\
    & = \mathcal{I}^{(\tau-1)} 
    + L\mu,
\end{align*}
    where the last inequality is due to $L$-smoothness of $\mathcal{F}_q$. We have, therefore, obtained
    \begin{align*}
        \mathcal{I}^{(\tau)} \le \mathcal{I}^{(\tau-1)} + L\mu.
    \end{align*}
    Using this result and combining it with~\eqref{eq: I(0)} yields
    \begin{align} \label{eq:I(t)}
        \mathcal{I}^{(\tau)} &\le \mathcal{I}^{(0)}  + \tau(L\mu) \nonumber \\ &= \left\| \nabla \mathcal{F}_q(\mathbf{w}^{(t)}) - \nabla \mathcal{F}(\mathbf{w}^{(t)}) \right\|_1 + \tau L\mu,
    \end{align}
    where we have returned to the original notation from which we temporarily deviated in~\eqref{eq: notation_switch}. Plugging~\eqref{eq:I(t)} back into~\eqref{eq:main_eq}, we get 
    \begin{align*}
        \mathbb{E}\{\Omega \mid \mathbf{w}^{(t)}\} &\le -T_E\|\nabla \mathcal{F}\|_1 + T_E\Big(\frac{2\sigma d}{\sqrt{B}}\Big) \\& \hspace{-5em}+ 2T_E\left\| \nabla \mathcal{F}_q(\mathbf{w}^{(t)}) - \nabla \mathcal{F}(\mathbf{w}^{(t)}) \right\|_1 + T_E(T_E-1)L\mu.
    \end{align*}
    With this result, we can bound~\eqref{eq: 1stterm_rewrite} after taking expectation
    \begin{align*}
    &\mathbb{E} \Big\{ \!\!- \mu\sum_{q=1}^Q   \frac{D_q}{N}\sum_{\tau=0}^{T_E-1}\Big\langle \nabla \mathcal{F}(\mathbf{w}^{(t)}),\,\operatorname{sgn}\!\big(\hat{\mathbf{g}}_{q}^{(t,\tau)}\big)\Big\rangle \Big\}
        \\ &\le \mu\Big(\!\!-T_E\|\nabla \mathcal{F}\|_1 + \frac{2\sigma dT_E}{\sqrt{B}}+ 2T_E\,\zeta + T_E(T_E-1)L\mu\Big).
    \end{align*}
    Combining this with the bounded second term, we get
    \begin{align*}
        \frac{\mathbb{E}\left\{\mathcal{F}(\mathbf{w}^{(t+1)}) - \mathcal{F}(\mathbf{w}^{(t)})\mid\mathbf{w}^{(t)}\right\}}{\mu T_E} \le -\|\nabla \mathcal{F}(\mathbf{w}^{(t)})\|_1 + C,
    \end{align*}
    where
\[
C \,=\, 2\zeta \,+\, \frac{2\sigma d}{\sqrt{B}} \,+\, \big(\frac{3T_E}{2}-1\big)L\mu.
\]
    Finally, we extend the expectation over the randomness in the process, apply a telescoping sum over the iterations, and rearrange to achieve the averaged performance bound
\begin{align*} 
\frac{1}{T_G}\!\sum_{t=0}^{T_G-1} 
\mathbb{E}\!\left\{\big\|\nabla \mathcal{F}(\mathbf{w}^{(t)})\big\|_{1}\right\}
\;\le\;
\frac{\mathcal{F}(\mathbf{w}^{(0)})-\mathcal{F}^\star}{\mu T_GT_E }
\;+\; C.
\end{align*}

\section{Proof of Theorem~\ref{thm:2nd}} \label{app:thm2}
To avoid repeating the entire proof, we apply only the necessary modifications to the proof of Theorem~\ref{thm:1st}. Since the devices now transmit the signs of the corrected gradients, the new probability of an incorrect sign estimate in \eqref{eq:Omega} will be
\begin{align*} 
\text{Pr}\left\{\operatorname{sgn}\!\big([\nabla \mathcal{F}]_i\big)\neq \operatorname{sgn}\!\big([\hat{\mathbf{g}}_q^{(\tau)}]_i+\rho[\boldsymbol{\delta}_q]_i\big)\right\} , 
\end{align*} 
where $\boldsymbol{\delta}_q=\boldsymbol{\delta}_q^{(t-1)}=\mathbf{c}^{(t-1)}-\mathbf{c}_q^{(t-1)}$. Following steps similar to those in \eqref{eq:main_eq}, we ultimately need to bound the following quantity:
\begin{align} 
&\sum_{i=1}^d\mathbb{E}\big\{\big|[{\mathbf{g}}_q^{(\tau)}]_i - [\nabla \mathcal{F}]_i+\rho[\boldsymbol{\delta}_q]_i\big|\big\}. \tag{\ensuremath{\clubsuit}} \label{eq:club} 
\end{align} 
Inserting the values from~\eqref{eq:edge_anchor}--\eqref{eq:global_anchor}, and adding and subtracting $[\nabla \mathcal{F}_q(\mathbf{w}^{(t)})]_i$, yields
\begin{align*} 
\eqref{eq:club} &= \sum_{i=1}^d \mathbb{E} \Big\{ \!\big| [{\mathbf{g}}_q^{(\tau)}]_i - [\nabla \mathcal{F}_q(\mathbf{w}^{(t)})]_i -[\nabla \mathcal{F}(\mathbf{w}^{(t)})]_i  \\ 
& + \rho[\nabla \mathcal{F}(\mathbf{w}^{(t-1)})]_i+ [\nabla \mathcal{F}_q(\mathbf{w}^{(t)})]_i- \rho[\nabla \mathcal{F}_q(\mathbf{w}^{(t-1)})]_i \big| \!\Big\}\\ 
&\le 2\rho\mathbb{E}\Big\{L\big \|\mathbf{w}^{(t)} - \mathbf{w}^{(t-1)} \big\|_\infty\Big\}+(1-\rho)\times\\
&\quad\,\sum_{i=1}^d \mathbb{E} \Big\{ \big| [\nabla \mathcal{F}_q(\mathbf{w}^{(t)})]_i - [\nabla \mathcal{F}(\mathbf{w}^{(t)})]_i\big|\Big\}\\
&+\sum_{i=1}^d \mathbb{E} \Big\{ \big| [{\mathbf{g}}_q^{(\tau)}]_i - [\nabla \mathcal{F}_q(\mathbf{w}^{(t)})]_i\big|\Big\} \\ 
&\le\! 2\rho L\mu T_E \!+\! (1-\rho)\zeta+\!\sum_{i=1}^d \mathbb{E} \Big\{ \!\big| [{\mathbf{g}}_q^{(\tau)}]_i - [\nabla \mathcal{F}_q(\mathbf{w}^{(t)})]_i\big|\!\Big\},
\end{align*} 
where we have used \eqref{eq: w_update} and the $L$-smoothness of the functions. Defining
\begin{align} 
\widetilde{\mathcal{I}}^{(\tau)} \triangleq \sum_{i=1}^d\mathbb{E}\big\{\big|[{\mathbf{g}}_q^{(\tau)}]_i - [\nabla \mathcal{F}_q(\mathbf{w}^{(t)})]_i\big|\big\}, 
\end{align} 
and applying the same recursive argument as before, we obtain
$\widetilde{\mathcal{I}}^{(\tau)} \le  \tau L\mu$ with
$\widetilde{\mathcal{I}}^{(0)}=0$.  Substituting this result into \eqref{eq:club} and continuing the analysis as in Appendix~\ref{app:thm1} completes the proof.

\section{Majority-Vote Error Analysis} \label{App_A}
In this section, we generalize the results of Theorems~\ref{thm:1st}--\ref{thm:2nd} to the case where each edge server manages a cluster of $M$ devices. We reiterate that the convergence of our proposed sign-based algorithms in the case of a single-device edge server primarily hinges on the probability of an incorrect sign estimate remaining bounded:
\begin{align} \label{eq:sinProb}
    P_e \triangleq \text{Pr}\left\{\operatorname{sgn}\big([\nabla \mathcal{F}]_i\big)\neq \operatorname{sgn}\big([\hat{\mathbf{g}}_q^{(\tau)}]_i+[\boldsymbol{\delta}_q]_i\big)\right\},
\end{align}
where $\boldsymbol{\delta}_q$ denotes an arbitrary gradient correction term, including $\boldsymbol{\delta}_q=0$. We further bounded $P_e$ using a relaxation argument and Markov's inequality:
\begin{align} \label{eq:bound again}
    P_e &\leq \frac{\mathbb{E}\left\{\big|[\hat{\mathbf{g}}_q^{(\tau)}]_i+[\boldsymbol{\delta}_q]_i - [\nabla \mathcal{F}]_i\big|\right\}}{\big|[\nabla \mathcal{F}]_i\big|} \nonumber\\
    &= \frac{\mathbb{E}\left\{\left|[\hat{\mathbf{g}}_q^{(\tau)}]_i - {[{\mathbf{g}}_q^{(\tau)}]_i} + {[{\mathbf{g}}_q^{(\tau)}]_i} +[\boldsymbol{\delta}_q]_i- [\nabla \mathcal{F}]_i\right|\right\}}{\big|[\nabla \mathcal{F}]_i\big|}\nonumber\\
    &\le \frac{\nu}{\big|[\nabla \mathcal{F}]_i\big|} + \frac{\mathbb{E}\left\{|\beta|\right\}}{\big|[\nabla \mathcal{F}]_i\big|} =: \psi, \nonumber
\end{align}
where \(\psi\) denotes the bound on $P_e$, $\beta = {[{\mathbf{g}}_q^{(\tau)}]_i} +[\boldsymbol{\delta}_q]_i - [\nabla \mathcal{F}]_i$, and $\nu^2 = \mathbb{E}\Big\{\big|[\hat{\mathbf{g}}_q^{(\tau)}]_i - {[{\mathbf{g}}_q^{(\tau)}]_i}\big|^2 \Big\}$.
\smallskip

For a given cluster with \(M=|\mathcal{V}^q|\) devices, the majority-vote error probability becomes
\begin{align*} %\label{eq:MajProb}
    P_e^{(M)} = \text{Pr}\Big\{\!\operatorname{sgn}\!\big([\nabla \mathcal{F}]_i\big)\neq\operatorname{sgn}\!\Big(\!\sum_{k \in \mathcal{V}^q} [\widetilde{\mathbf{s}}_{qk}^{(\tau)}]_i\Big)\Big\},
\end{align*}
where $\widetilde{\mathbf{s}}_{qk}^{(\tau)}=\operatorname{sgn}(\hat{\mathbf{g}}_{qk}^{(\tau)}+\boldsymbol{\delta}_q)$.
Similar to \cite{bernstein2018signsgd}, we argue that this probability is bounded by the same threshold as $P_e$ is. In other words, if $P_e \le \psi$, then $P_e^{(M)}\le \psi$. There are two cases to consider: when $\psi \ge 1$, the inequality $P_e^{(M)} \le \psi$ is immediate; when $\psi < 1$, more careful examination is required. Without loss of generality assume $\operatorname{sgn}([\nabla \mathcal{F}]_i) = -1$ and let us define
\begin{align*}
    X \triangleq {[\hat{\mathbf{g}}_q^{(\tau)}]_i} +[\boldsymbol{\delta}_q]_i.
\end{align*}
From \eqref{eq:sinProb}, we have
\begin{align*} %\label{eq:Cantelli}
    P_e &= \text{Pr}\left\{X > 0\right\} =  \text{Pr}\big\{X-\mathbb{E}\{X\} > -\mathbb{E}\{X\}\big\}.
\end{align*}
Now, note that
\begin{align*}
    - \mathbb{E}\{X\}&= -\mathbb{E}\{\beta\}-[\nabla \mathcal{F}]_i = -\mathbb{E}\{\beta\}+\big|[\nabla \mathcal{F}]_i\big| \\ 
    &\ge -\mathbb{E}\{|\beta|\}+\big|[\nabla \mathcal{F}]_i\big| \, > \, \nu \ge 0,
\end{align*}
where the last inequality is simply due to $\psi < 1$. Hence, $-\mathbb{E}\{X\} > 0$ and we can use Cantelli's inequality~\cite{BoucheronLugosiMassart2013} to obtain a tight one-sided tail bound:
\begin{align*}
    P_e &=  \text{Pr}\big\{X-\mathbb{E}\{X\} > -\mathbb{E}\{X\}\big\} \\
    &\le \frac{\nu^2}{\nu^2+\big(\mathbb{E}\{X\}\big)^2} < \frac{\nu^2}{\nu^2+\nu^2} = \frac{1}{2}.
\end{align*}
Therefore, we have obtained $P_e < 1/2$ when $\psi < 1$ (the same analysis holds when $[\nabla \mathcal{F}]_i$ is positive). As we shall see, this result is crucial for the final step of the argument.

We now draw an analogy to our problem by considering a fundamental channel coding scenario. Suppose a single information bit (0 or 1) is repeatedly transmitted over a noisy channel. In this setting, the maximum a posteriori (MAP) detector, optimal for minimizing the probability of detection error \cite{kay_estimation}, reduces to a simple majority-vote rule applied to the received samples as the following illustrates.
Let $b$ be the bit sent, repeated $M$ times. Under the conditional independence assumption on the device-level sign errors, the channel flips each bit independently with probability $P_e$. Let $r$ be what we receive, and let $n$ be the number of $1$'s in $r$.

\begin{itemize}
    \item If $b = 0$ was sent, the number of errors equals $n$:
    \[
    \text{Pr}\{r \mid b=0\} = P_e^{n} (1-P_e)^{M-n}.
    \]
    \item If $b = 1$ was sent, the number of errors equals $M-n$:
    \[
    \text{Pr}\{r \mid b=1\} = P_e^{M-n} (1-P_e)^{n}.
    \]
\end{itemize}
With an equiprobable prior on $b$, the MAP detector chooses the $b$ with larger likelihood:
\[
\frac{\text{Pr}\,\{r \mid b=1\}}{\text{Pr}\,\{r \mid b=0\}}
= \left( \frac{1-P_e}{P_e} \right)^{2n - M}.
\]
Since we have established that $P_e < 1/2$, the ratio is $> 1$ exactly when $n>M/2$. This is precisely \emph{majority-vote}.
For the case $n = M/2$, the likelihood ratio equals $1$, meaning that the MAP rule has no preference between $0$ and $1$ and may choose either value, for example by random tie-breaking.
We reiterate that MAP is optimal, meaning that no alternative decoding rule can achieve a smaller probability of error than majority-vote, including the crude decoder that examines only the first received bit and outputs it. In other words, $P_e^{(M)} \le P_e$, and the bounds in Theorems~\ref{thm:1st}--\ref{thm:2nd}, which are looser than $\psi$, also hold for the majority-vote case.

\bibliographystyle{IEEEtran}
\bibliography{IEEEabrv,Ref}

\end{document}